\documentclass[a4paper, twoside]{amsart}
\usepackage{mathrsfs, tabularx}

\usepackage{graphicx}
\usepackage {amsthm, epsfig, amsmath, amssymb, amscd, float, latexsym}

\theoremstyle{plain}
\newtheorem{theorem}{Theorem}[section]
\newtheorem{proposition}[theorem]{Proposition}

\newtheorem{lemma}[theorem]{Lemma}
\theoremstyle{definition}
\newtheorem{definition}[theorem]{Definition}

\newcommand {\Set}[1] {\mathbb{#1}}
\newcommand{\setR}[0]{\Set{R}}

\newcommand{\slaz}[0]{{\setminus\{0\}}}

\newcommand {\proofread}[1]{ }

\newcommand{\comment}[1]{}

\newcommand{\pd}[2]{\frac{\partial #1}{\partial #2}}

\newcommand{\kappaI}[0]{^{(1)}\!\kappa}
\newcommand{\kappaII}[0]{^{(2)}\!\kappa}
\newcommand{\kappaIII}[0]{^{(3)}\!\kappa}

\newcounter{saveenum}

\newcommand{\HOX}[1]{}

\title{A restatement of the normal form theorem for area metrics}
\author[Dahl]{Matias F. Dahl} 
\address{
Matias Dahl\\  
Aalto University\\
Mathematics\\
P.O. Box 11100\\
FI-00076 Aalto\\
Finland
}

\urladdr{http://www.math.tkk.fi/\textasciitilde{}fdahl/}

\date{\today}

\begin{document}
\begin{abstract}
  An \emph{area metric} is a $0\choose 4$-tensor with certain
  symmetries on a $4$-manifold that represent a non-dissipative linear
  electromagnetic medium. A recent result by Schuller, Witte and
  Wohlfarth provides a pointwise normal form theorem for such area
  metrics. This result is similar to the Jordan normal form theorem
  for $1\choose 1$-tensors, and the result shows that any area metric
  belongs to one of $23$ metaclasses with explicit coordinate
  expressions for each metaclass. In this paper we restate and prove
  this result for skewon-free $2\choose 2$-tensors and show that in general, each
  metaclasses has three different coordinate
  representations, and each of metaclasses I, II, $\ldots$, VI, VII need only one 
coordinate representation.
\end{abstract}

%
%
%

\maketitle

\section{Introduction}
An \emph{area  metric} on a $4$-manifold
 $N$ is a $0\choose
4$-tensor $G$ on $N$ that gives a symmetric (possibly indefinite) inner product
for bivectors on $N$. 
The motivation for studying area metrics is that they appear as a
natural generalisation of Lorentz metrics in physics. For example, in
relativistic electromagnetics, a Lorentz metric always describes an
isotropic medium, but using an area metric one can also model
anisotropic medium, where differently polarised waves can propagate
with different wave-speeds. 
Area metrics also appear when studying the propagation of a photon in
a vacuum with a first order correction from quantum electrodynamics
\cite{DruHath:1980, Schuller:2010}.
The Einstein field equations have also been generalised
into equations where the unknown field is an area metric
\cite{PSW_JHEP:2007}.
For further examples, see \cite{PunziEtAl:2009, Schuller:2010}, and
for the differential geometry of area metrics, see 
\cite{SchullerWohlfarth:2006, PSW_JHEP:2007}.

The present work is motivated by a recent result by
Schuller, Witte and Wohlfarth \cite{Schuller:2010} which is a 
normal form theorem for area metrics on a $4$-manifold $N$.
Essentially, this theorem states that there are $23$ normal forms for
area metrics, 
and if $G$ is any area metric on $N$ and $p\in N$, one can find
coordinates around $p$ such that $G\vert_p$ is one of the normal forms
(up to simple operations) \cite[Theorem 4.3]{Schuller:2010}.
What is more,  16 of the metaclasses are
unphysical in the sense that Maxwell's equations are
not well-posed in these metaclasses. This leaves only 7 metaclasses
that can describe physically relevant electromagnetic medium \cite{Schuller:2010}.
The importance of this result is that in arbitrary coordinates an area
metric depends on $21$ real numbers, but each normal form depend
on at most $6$ real numbers and $3$ signs $\pm 1$.  This
reduction in variables has proven particularly useful when studying properties of
the Fresnel equation (or dispersion equation) for a propagating
electromagnetic wave \cite {Schuller:2010, FavaroBergamin:2011}. 
Namely, without assumptions on either the area
metric or the coordinates, the Fresnel equation usually leads to
algebraic expressions that are quite difficult to manipulate, even
with computer algebra \cite{Dahl:2011:Closure}.

In addition to area metrics, there are multiple other ways to model
the medium in (relativistic) electrodynamics. Another common formalism
is the so called \emph{pre-metric} formulation, where the medium is
modelled by an antisymmetric $2\choose 2$-tensor $\kappa$ on a
$4$-manifold $N$.  In this formalism, an electromagnetic medium
$\kappa$ is pointwise determined by $36$ real numbers \cite{Obu:2003}.
Under suitable conditions it follows that the area-metrics on $N$ are
in one-to-one correspondence with invertible skewon-free $2\choose
2$-tensors on $N$.  (See \cite{FavaroBergamin:2011} and Propositions
\ref{theorem:Decomp} and \ref {prop:correspondence} below).
Because of this correspondence, the normal form theorem in
\cite{Schuller:2010} can, of course, be stated also for skewon-free
$2\choose 2$-tensors.
The contribution of this paper we write down this restatement
explicitly, and also prove the result in this setting by following the
proof in \cite{Schuller:2010}.
Below, this is given by Theorem
\ref{thm:classification}.
However, we obtain a slightly different result. In
\cite{Schuller:2010}, area metrics divide into 23 metaclasses and each
metaclass has two representations in local coordinates, but in Theorem
\ref{thm:classification}, we obtain three different coordinate
representations for each metaclass. Moreover, for metaclasses I, II,
$\ldots$, VI, VII we show that only one coordinate representation is
needed per metaclass.


 A minor difference is also that in Theorem \ref{thm:classification},
 one does not need to assume that $\kappa$ is invertible. This was
 already noted in \cite{FavaroBergamin:2011}.

 This paper relies on computations by computer algebra. Mathematica
 notebooks for these computations can be found on the author's
 homepage.

 \section{Maxwell's equations}
 \label{mainSec}
 By a \emph{manifold} $M$ we mean a second countable topological Hausdorff
 space that is locally homeomorphic to $\setR^n$ with $C^\infty$-smooth
 transition maps. All objects are assumed to be smooth and real where defined.  
 Let $TM$ and $T^\ast M$ be the tangent and cotangent bundles,
 respectively, and for $k\ge 1$, let $\Lambda^k(M)$ be the set of
 antisymmetric $k$-covectors, so that $\Lambda^1(N)=T^\ast N$.  
 Also, let $\Lambda_k(M)$ be the set of antisymmetric $k$-vectors.
 Let $\Omega^k_l(M)$
 be $k\choose l$-tensors that are antisymmetric in their $k$ upper
 indices and $l$ lower indices. In particular, let $\Omega^k(M)$ be the
 set of $k$-forms. 
  Let $C^\infty(M)$ be the set of functions. 
 %
 %
 The Einstein summing convention is used throughout. When writing tensors 
 in local coordinates we assume that the components satisfy the same symmetries as
 the tensor. 

 \subsection{Maxwell's equations on a $4$-manifold}
 \label{sec:MaxOn4}
 \HOX{Section OK 5.5}
 Suppose $N$ is a $4$-manifold. 
 On a $4$-manifold $N$, Maxwell's equations read
 \begin{eqnarray}
 \label{max4A}
 dF &=& 0, \\
 \label{max4B}
 dG &=& j,
 \end{eqnarray}
 where $d$ is the exterior derivative on $N$, $F,G\in \Omega^2 (N)$, 
 and $j\in \Omega^3(N)$. 
 %
 By an \emph{electromagnetic medium} on $N$ 
 we mean a map 
 \begin{eqnarray*}
    \kappa \colon \Omega^2(N) &\to& \Omega^2(N).
 \end{eqnarray*} 
 We then say that $2$-forms $F,G\in \Omega^2(N)$ \emph{solve Maxwell's
   equations in medium $\kappa$} if  $F$ and $G$ satisfy equations
 \eqref{max4A}--\eqref{max4B} and
 \begin{eqnarray}
 \label{FGchi}
   G &=& \kappa(F).
 \end{eqnarray}
 Equation \eqref{FGchi} is known as the \emph{constitutive equation}.
 If $\kappa$ is invertible, 
 it follows that one can eliminate half of the free variables in
 Maxwell's equations \eqref{max4A}--\eqref{max4B}.
 We assume that $\kappa$ is linear and determined pointwise so that we can represent
 $\kappa$ by an antisymmetric $2\choose 2$-tensor $\kappa \in
 \Omega^2_2(N)$. If in coordinates $\{x^i\}_{i=0}^3$ for $N$ we have
 \begin{eqnarray}
 \label{eq:kappaLocal}
   \kappa &=& \frac 1 2 \kappa^{ij}_{lm} dx^l\otimes dx^m\otimes \pd{}{x^i}\otimes \pd{}{x^j}
 \end{eqnarray}
 and $F = F_{ij} dx^i \otimes dx^j$ and $G = G_{ij} dx^i \otimes dx^j$, 
 then constitutive equation \eqref{FGchi} reads
 \begin{eqnarray}
  \label{FGeq_loc}
    G_{ij} &=& \frac 1 2 \kappa_{ij}^{rs} F_{rs}.
 \end{eqnarray}

 \subsection{Decomposition of electromagnetic medium}
 \label{media:decomp}
 \HOX{Section OK 5.5}
 Let $N$ be a $4$-manifold. Then
 at each point on $N$, a general antisymmetric $2\choose
 2$-tensor depends on $36$ parameters. Such tensors canonically decompose
 into three linear subspaces. The motivation for this decomposition is
 that different components in the decomposition enter in different
 parts of electromagnetics.  See \cite[Section D.1.3]{Obu:2003}.  
 The below formulation is taken from \cite{Dahl:2009}.

 If $\kappa\in \Omega^2_2(N)$ we define the
 \emph{trace} of $\kappa$ as the smooth function $N\to \setR$ given by
 $\operatorname{trace} \kappa = \frac 1 2 \kappa_{ij}^{ij}$
 when $\kappa$ is locally given by equation \eqref{eq:kappaLocal}.
 Writing 
 $\operatorname{Id}$ as in equation \eqref{eq:kappaLocal} gives
 $\operatorname{Id}^{ij}_{rs}= \delta^i_r\delta^j_s-\delta^i_s\delta^j_r$, so 
 $\operatorname{trace}\operatorname{Id} = 6$ when $\dim N=4$. 

 \begin{proposition}[Decomposition of a  $2\choose 2$-tensors]
 \label{theorem:Decomp}
 Let $N$ be a $4$-manifold, and let
 \begin{eqnarray*}
 Z &=& \{ \kappa \in \Omega^2_2(N) : u\wedge \kappa(v) = \kappa(u)\wedge v \,\,\mbox{for all}\,\, u,v\in \Omega^2(N),\\
 & & \quad\quad\quad\quad\quad\quad \operatorname{trace} \kappa = 0\},\\
 W &=& \{ \kappa \in \Omega^2_2(N) : u\wedge \kappa(v) = -\kappa(u)\wedge v \,\,\mbox{for all}\,\, u,v\in \Omega^2(N)\} \\ 
  &=& \{ \kappa \in \Omega^2_2(N) : 
  u\wedge \kappa(v) = -\kappa(u)\wedge v \,\,\mbox{for all}\,\, u,v\in \Omega^2(N), \\
 & & \quad\quad\quad\quad\quad\quad \operatorname{trace} \kappa = 0 \}, \\
 U &=& \{ f \operatorname{Id}\in \Omega^2_2(N) : f\in C^\infty(N) \}.
 \end{eqnarray*}
 Then
 \begin{eqnarray}
 \label{AdecompSet}
  \Omega^2_2(N) &=& Z\,\,\oplus\,\, W \,\,\oplus\,\, U,
 \end{eqnarray}
 and pointwise, $\dim Z = 20$,  $\dim W = 15$ and  $\dim U = 1$.
 \end{proposition}

 If we write a $\kappa\in \Omega^2_2(N)$ as
 \begin{eqnarray*}
   \kappa &=& \kappaI \,\,+ \,\,\kappaII\,\,+\,\,\kappaIII
 \end{eqnarray*}
 with $\kappaI\in Z$, $\kappaII\in W$, $\kappaIII\in U$, then we say that 
 $\kappaI$ is the \emph{principal part},
 $\kappaII$ is the \emph{skewon part},
 $\kappaIII$ is the \emph{axion part} of $\kappa$.


 \subsection{Representing $\kappa$ as a $6\times 6$ matrix}
 \label{sec:Rep6x6}
 \HOX{Section 2.3 OK: 9.5.2011}
 Let $O$ be the ordered set of index pairs $\{ 01, 02, 03$, $23, 31, 12\}$. 
 If $I\in O$, let us also denote the individual indices by $I_1$ and $I_2$. Say,
 if $I=31$ then $I_2=1$. 

 If $\{x^i\}_{i=0}^3$ are local coordinates for a $4$-manifold $N$, and
 $J\in O$ we define $dx^J = dx^{J_1}\wedge dx^{J_2}$.  
 A basis for $\Omega^2(N)$
 is given by $\{ dx^{J}: J\in O\}$, that is,
 \begin{eqnarray}
 \label{eq:2basis}
 \{ dx^0\wedge dx^1, dx^0\wedge dx^2, dx^0\wedge dx^3, dx^2\wedge dx^3, dx^3\wedge dx^1, dx^1\wedge dx^2\}.
 \end{eqnarray}
 This choice of basis follows \cite[Section A.1.10]{Obu:2003} and
 \cite{FavaroBergamin:2011}.

 If $\kappa\in \Omega^2_2(N)$ is written as in equation
 \eqref{eq:kappaLocal} and $J\in O$, then
 \begin{eqnarray}
 \label{eq:kappaMatDef}
 \kappa(dx^{J}) = \sum_{I\in O} \kappa^J_I dx^I = \kappa^J_I dx^I,\quad J \in O,
 \end{eqnarray}
 where $\kappa^J_I = \kappa^{J_1 J_2}_{I_1 I_2}$ and in the last
 equality we have extended the Einstein summing convention also to
 elements in $O$. We will always use capital letters $I,J,K,\ldots$ to
 denote elements in $O$.

 Let $b$ be the natural bijection $b\colon O\to \{1,\ldots, 6\}$.  Then
 coefficients $\{\kappa^J_I : I,J\in O\}$ can be identified with a
 $6\times 6$ matrix.
 To do this
 identification in a systematic way, we denote by $(m(I,J) )_{IJ}$ the
 $6\times 6$ matrix whose entry at row $b(I)\in \{1,\ldots, 6\}$ and
 column $b(J)\in \{1, \ldots, 6\}$ is given by expression $m(I,J)$. For
 example, if $A$ is the $6\times 6$ matrix $A=(\kappa^J_I)_{IJ}$, then
 \begin{eqnarray}
 \label{eq:matrixCorr}
 \kappa^J_I &=& A_{b(I) b(J)},\quad I,J\in O.
 \end{eqnarray}
 If $\eta\in\Omega^2_2(N)$ and $B=(\eta^J_I)_{IJ}$, where $\eta^J_I$ represent
 $\eta\vert_p$ as in equation \eqref{eq:kappaMatDef}, then 
 $
 ( (\kappa\circ \eta)^J_I)_{IJ}= A B$
 .
 This compatibility with the matrix multiplication is the motivation for
 using the matrix representation for $\kappa$ as in 
 equation \eqref{eq:matrixCorr} (and not the transpose of $A$).

 Suppose $\{x^i\}_{i=0}^3$ and $\{\widetilde x^i\}_{i=0}^3$ are
 overlapping coordinates, and suppose that in these coordinates
 $\kappa$ is represented by $\kappa^J_I$ and $\widetilde \kappa^J_I$ as
 in equation \eqref{eq:kappaMatDef}.  Then we have the transformation
 rule
 \begin{eqnarray}
 \label{eq:kappaTransRule}
 \widetilde \kappa^J_I &=&  \pd{\widetilde x^J}{x^K} \kappa^K_L   \pd{x^L}{\widetilde x^I}, \quad I,J\in O,
 \end{eqnarray}
 where 
 \begin{eqnarray*}
 \pd{\widetilde x^I}{x^J} &=& \pd{\widetilde x^{I_1}}{x^{J_1}}\pd{\widetilde x^{I_2}}{x^{J_2}}-\pd{\widetilde x^{I_2}}{x^{J_1}}\pd{\widetilde x^{I_1}}{x^{J_2}}, \quad I,J\in O
 \end{eqnarray*}
 and $\pd{x^I}{\widetilde x^J}$ is defined analogously by exchanging $x$ and $\widetilde x$.
 For $I,J\in O$, we then have $\pd{\widetilde x^I}{x^K} \pd{x^K}{\widetilde x^J}
 =\delta^I_J$, where
 $
 \delta^I_J = \delta^{I_1}_{J_1}\delta^{I_2}_{J_2}-\delta^{I_2}_{J_1}\delta^{I_1}_{J_2}. 
 $ 


 \subsection{The Hodge star operator} 
 \label{sec:Hodge}
 \HOX{Section 2.5 OK: 9.5.2011}
 By a \emph{pseudo-Riemann metric} on a manifold $M$ we mean a
 symmetric $0\choose 2$-tensor $g$ that is non-degenerate. If $M$
 is not connected we also assume that $g$ has constant signature. If
 $g$ is positive definite we say that $g$ is a \emph{Riemann
 metric}. 

 Suppose $g$ is a pseudo-Riemann metric on an orientable manifold $M$
 with $n=\dim M\ge 1$.  For $p\in\{0,\ldots, n\}$, the \emph{Hodge star
 operator} $\ast$ is the linear map $\ast\colon \Omega^p(M)\to
 \Omega^{n-p}(M)$ defined as \cite[p. 413]{AbrahamMarsdenRatiu:1988}
 \begin{eqnarray*}
 \label{hodgedef}
 \ast(dx^{i_1} \wedge \cdots \wedge dx^{i_p}) &=& \frac{\sqrt{|\det g|}}{(n-p)!} g^{i_1 l_1}\cdots  g^{i_p l_p} \varepsilon_{l_1 \cdots l_p\, l_{p+1} \cdots l_n} dx^{l_{p+1}}\wedge \cdots \wedge  dx^{l_{n}},
 \end{eqnarray*}
 where $x^i$ are local coordinates in an oriented atlas,
 $g=g_{ij}dx^i\otimes dx^j$, $\det g= \det g_{ij} $, $g^{ij}$ is the
 $ij$th entry of $(g_{ij})^{-1}$, and $\varepsilon_{l_1\cdots l_n}$ is
 the \emph{Levi-Civita permutation symbol}. We treat
 $\varepsilon_{l_1\cdots l_n}$ as a purely combinatorial object (and
 not as a tensor density). We also define $\varepsilon^{l_1\cdots l_n}=
 \varepsilon_{l_1\cdots l_n}$.

 If $g$ is a pseudo-Riemann metric on an oriented
 $4$-manifold $N$, then the Hodge star operator for $g$ induces a
 $2\choose 2$-tensor $\kappa=\ast_g\in \Omega^2_2(N)$. If $\kappa$ is
 written as in equation \eqref{eq:kappaLocal} for local coordinates
 $x^i$ then
 \begin{eqnarray}
 \label{eq:hodgeKappaLocal}
 \kappa^{ij}_{rs} &=& \sqrt{\vert g\vert} g^{ia}g^{jb} \varepsilon_{abrs}
 \end{eqnarray}
 and $\kappa$ has only a principal part. See for example \cite[Proposition 2.2]{Dahl:2011:Closure}.

 Next we show that two pseudo-Riemann metrics can be combined by
 conjugation into a third pseudo-Riemann metric. 

 \begin{proposition}
 \label{prop:lorentzMixing}
 Suppose $g$ and $h$ are  pseudo-Riemann metrics 
 on an orientable $4$-dimensional manifold $N$. Then the
 pseudo-Riemann metric $k$ defined as 
 \begin{eqnarray*}
 \label{eq:kDef}
 k&=& g_{ia} h^{ab} g_{bj} dx^i\otimes dx^j
 \end{eqnarray*}
 satisfies
 \begin{eqnarray*}
 \label{eq:MixDefEqB}
 \ast_{k} &=& \operatorname{sgn}\left(\frac{\det g}{\det h}\right) \ast_{g}^{-1} \circ \ast_{h} \circ \ast_{g}.
 \end{eqnarray*}
 Conversely, if $\widetilde k$ is a pseudo-Riemann metric such
 that $\ast_{\widetilde k} = \lambda \ast_{g}^{-1} \circ \ast_{h} \circ \ast_{g}$ for some 
 $\lambda\in C^\infty(N)$, then $k$ and $\widetilde k$ are in the same conformal class.
 \end{proposition}

 \begin{proof}
 Let $g_{ij}, h_{ij}$ and $k_{ij}$ be components for $g,h$ and $k$, respectively.
 Using $\varepsilon_{ijkl} A^{ia} A^{jb} A^{kc} A^{ld} = \varepsilon^{abcd} \det A$ we obtain
 obtain
 \begin{eqnarray*}
 \ast_k( dx^{i}\wedge dx^j) &=& \operatorname{sgn}(\det g) \frac{\vert\det h\vert^{-1/2}} 2 g^{ia} g^{jb} h_{ac} h_{bd}\varepsilon^{cdrs} g_{ru} g_{sv} dx^u\wedge dx^v
 \end{eqnarray*}
 Similarly writing out $\ast_g^{-1}\circ \ast_h\circ \ast_g$ gives the first claim.
 The second claim follows by the lemma below.
 \end{proof}

 The next lemma is a slight generalisation of Theorem 1 in \cite{Dray:1989}.

 \begin{lemma}
 \label{eq:hodgeConformal}
 Suppose $g$ and $h$ are pseudo-Riemann metrics on an orientable
 $4$-dimensional manifold $N$.  If $\ast_g = f \ast_h$ for some $f\in
 C^\infty(N)$, then $f=1$ and $g=\lambda h$ for some $\lambda \in
 C^\infty(N)$.
 \end{lemma}

 \begin{proof}
  Since we only need to prove the claim at one point, let
  $\{x^i\}_{i=0}^3$ be coordinates for a connected neighbourhood $U$
  around some $p\in N$ where $h\vert_p$ is diagonal with entries $\pm
  1$.
 Squaring $\ast_g = f \ast_h$ gives $f^2 = 1$, so in $U$ we have either
 $f=1$ or $f=-1$. 
 By equation \eqref{eq:hodgeKappaLocal},
 \begin{eqnarray}
 \label{eq:gghh}
 \sqrt{\vert \det g\vert} g^{ia} g^{jb} \varepsilon_{abrs} &=& f   \sqrt{\vert \det h\vert} h^{ia} h^{jb} \varepsilon_{abrs}.
 \end{eqnarray}
 Contracting by $\varepsilon^{mnrs}$ and using equation \eqref{eq:varEpsContraction} gives
 \begin{eqnarray}
 \label{eq:gghh1}
 \sqrt{\vert \det g\vert} \left( g^{ij} g^{kl} -g^{ik} g^{jl} \right)
 &=& f   \sqrt{\vert \det h\vert} \left( h^{ij} h^{kl} -h^{ik} h^{jl} \right)
 \end{eqnarray}
 for all $i,j,k,l\in \{0,1,2,3\}$. Thus,
 if we have
 neither [$i=j$ and $k=l$] nor [$i=k$ and $j=l$], then
 \begin{eqnarray}
 \label{eq:gComm123}
 g^{ij} g^{kl} &=&g^{ik} g^{jl}. 
 \end{eqnarray}
 Thus, if $i,j,k,l$ are distinct, then 
 \begin{eqnarray}
 \nonumber
 \left(  g^{ii}   g^{jj}- (g^{ij})^2\right) (g^{kl})^2 &=& 
 \left( g^{ii} g^{kl}\right)
 \left( g^{jj} g^{kl}\right)
 -
 \left( g^{ij} g^{kl}\right)
 \left( g^{ij} g^{lk}\right)\\
 &=&
 \nonumber
 \left( g^{ik} g^{il}\right)
 \left( g^{jk} g^{jl}\right)
 -
 \left( g^{ik} g^{jl}\right)
 \left( g^{il} g^{jk}\right)\\
 &=& 0,
 \label{eq:gggComm}
 \end{eqnarray}
 Combining equations \eqref{eq:gghh1} and \eqref{eq:gggComm} gives
 $(h^{ii} h^{jj} - (h^{ij})^2) (g^{kl})^2 = 0$. Hence $g$ is also diagonal at $p$.
 Equation \eqref{eq:gghh1} gives
 \begin{eqnarray} 
 \label{eq:manualGB}
 \sqrt{\vert \det g\vert}  g^{ii} g^{jj}  &=& f   \sqrt{\vert \det h\vert} h^{ii} h^{jj}, \quad i<j.
 \end{eqnarray}
 Writing out equation \eqref{eq:manualGB} for cases $(i,j)= (0,3),
 (1,3)$ and $(0,1)$ and multiplying the first two equations gives $
 \sqrt{\vert \det g\vert} (g^{33})^2 = f \sqrt{\vert \det h\vert}
 (h^{33})^2$.  Thus $f=1$ in $U$ and $g^{33} = \sigma h^{33}$ for some
 $\sigma \in \{\pm \left(\frac{\vert \det h\vert}{\vert
    \det g\vert}\right)^{1/4}\}$.  Setting $j=3$ in equation
 \eqref{eq:manualGB} then gives $g^{ii} = \sigma h^{ii}$ for $i\in \{0,1,2\}$.
 \end{proof}

 \subsection{Area metrics}
 \HOX{Section 2.4 OK}
 As described in the introduction, 
 an area metric is a geometry that at each point $p$ gives a
 (possible indefinite) inner product for bivectors, that is, for
 elements in $\Lambda_2(N)\vert_p$. 
 In this section we show  that area
 metrics are essentially in one-to-one correspondence with skewon-free
 $2\choose 2$-tensors.

 \begin{definition}
 Suppose $N$ is a $4$-manifold. An \emph{area metric}
 is a $0\choose 4$-tensor $G$ on $N$ such that
 \begin{enumerate}
 \item $G(u,v,p,q)$ is antisymmetric in $u,v$,
 \item For each $p\in N$, the quadratic form
 \begin{eqnarray*}
   \Lambda_2(N)\vert_p\,\times\, \Lambda_2(N)\vert_p&\to& \setR
 \end{eqnarray*}
 determined by 
 \begin{eqnarray*}
    (a\wedge b, u\wedge v )&\mapsto& G(a,b,u,v), \quad a,b,u,v\in \Lambda^1_p(N)
 \end{eqnarray*}
 is symmetric and non-degenerate.
 \end{enumerate}
 \end{definition}

 Suppose $G$ is a $0\choose 4$-tensor on a $4$-manifold. Then in local coordinates $\{x^i\}_{i=0}^3$, 
 \begin{eqnarray*}
  G = G_{ijkl} dx^i \otimes dx^j \otimes dx^k \otimes dx^l,
 \end{eqnarray*}
 and $G$ is an area metric if and only if components $G_{ijkl}$ satisfy \emph{(i)}
 $G_{ijrs} = -G_{jirs}$, \emph{(ii)} $G_{ijrs} = G_{rsij}$ and \emph{(iii)} the $6\times 6$
 matrix $(G_{I_1 I_2 J_1 J_2})_{IJ}$ is invertible.

 \begin{proposition}
 \label{prop:correspondence}
 Suppose $N$ is an orientable $4$-manifold and 
 $g$ is a pseudo-Riemann metric on $N$.
 Let
 $A$ be the map that maps a $\kappa\in \Omega^2_2(N)$ into the 
 $0\choose 4$-tensor 
 \begin{eqnarray*}
 A(\kappa)
 &=& \frac 1 2 \vert \det g\vert^{1/2} \kappa^{rs}_{ij} \varepsilon_{rskl} dx^i\otimes dx^j\otimes dx^k\otimes dx^l,
 \end{eqnarray*}
 where $\kappa^{ij}_{kl}$ are defined as in equation \eqref{eq:kappaLocal}. Then
 $\kappa\mapsto A(\kappa)$ is invertible,
 and if $\kappa$ is invertible as a linear map $\kappa\colon \Omega^2(N)\to \Omega^2(N)$, 
 then
 $\kappa$ is skewon-free if and only if $A(\kappa)$ is an area metric.
 \end{proposition}

 \begin{proof}
 A direct computation shows that $A(\kappa)$ is a tensor, and the identity
 \begin{eqnarray}
 \label{eq:varEpsContraction}
 \varepsilon^{ijkl}\varepsilon_{ijrs} = 2(\delta^k_r \delta^l_s - \delta^k_s \delta^l_r )
 \end{eqnarray}
 shows that $A$ is invertible. Using equation
 \eqref{eq:varEpsContraction} we also see that for the last equivalence
 we only need to show that $(\kappa^{rs}_{I_1 I_2} \varepsilon_{rs J_1
  J_2})_{IJ}$ is an invertible $6\times6$ matrix.  The result follows
 since the summation over $r,s$ can be written as a matrix
 multiplication.
 \end{proof}

 \section{The normal form theorem restated for  $2\choose 2$-tensors}
 \label{sec:classification22}
 In this section we formulate Theorem \ref{thm:classification}, which
 provides the restatement of the normal form theorem in
 \cite{Schuller:2010}. 
 First we introduce some terminology and
 notation.  Suppose $L\colon V\to V$ is a linear map where $V$ is a
 real $n$-dimensional vector space.  If the matrix representation of
 $L$ in some basis is $A\in \setR^{n\times n}$ and $A$ is written as in
 Theorem \ref{thm:jordan}, then we say that $L$ has \emph{Segre type}
 $\left[m_1\cdots m_r\, k_{1}\overline{k_{1}}\cdots
  k_{s}\overline{k_{s}}\right]$.  Moreover, by Theorem
 \ref{thm:jordan}, the Segre type depends only on $L$ and not on the
 basis.
 If $\kappa\in \Omega^2_{2}(N)\vert_p$, where $N$ is a $4$-manifold and
 $p\in N$, then we say that the Segre type of $\kappa\vert_p$ is the
 Segre type of the linear map $\kappa\vert_p\colon \Omega^2(N)\vert_p\to
 \Omega^2_p(N)\vert_p$.  By counting how many ways a $6\times 6$ matrix can be
 partitioned into blocks as in equation \eqref{eq:jordanDecomp}, we see
 that there are 23 possible Segre types for $\kappa\vert_p$. These are
 the Segre types listed in Theorem \ref{thm:classification}, that is,
 \begin{eqnarray}
 \label{eq:SegreTypes}
 [1\overline 1\, 1\overline 1\, 1\overline 1],\quad
 [2\overline 2\, 1\overline 1],\quad
 [3\overline 3],\quad 
 \cdots\quad \quad
 [321], \quad
 [31\, 1\overline 1],\quad
 [31\, 11]. 
 \end{eqnarray}

 To formulate Theorem \ref{thm:classification} we need the following definition. 
 \begin{definition}
 \label{def:simNotation}
 Suppose $N$ is a $4$-dimensional manifold, 
 $\kappa\in \Omega^2_2(N)$, $p\in N$ and 
 $V\in \setR^{6\times 6}$. We then write
 \begin{eqnarray}
 \label{eq:simNotation}
 \kappa\vert_p &\sim& V
 \end{eqnarray} 
 to indicate that there exist coordinates $\{x^i\}_{i=0}^3$ around $p$  for which  
 at least one of the below conditions is satisfied:
 \begin{enumerate}
 \item 
 \label{def:Sim:A}
   $(\kappa^J_I)_{IJ}=V$.
 \item 
 \label{def:Sim:B} 
 For the Riemann metric 
 $g = \operatorname{diag}(1,1,1,1)$ in coordinates $\{x^i\}_{i=0}^3$ 
 we have
 \begin{eqnarray*}
 ((\ast_g\circ \kappa\circ \ast_g)^J_I)_{IJ}&=&V.
 \end{eqnarray*}
 \item 
 \label{def:Sim:C} 
 For the pseudo-Riemann metric 
 $g = \operatorname{diag}(1,-1,-1,1)$ in coordinates $\{x^i\}_{i=0}^3$ 
 we have
 \begin{eqnarray*}
 ((\ast_g\circ \kappa\circ \ast_g)^J_I)_{IJ}&=&V.
 \end{eqnarray*}
 \end{enumerate}
 In the above $\eta^J_I$ denote the components as in equation
 \eqref{eq:kappaMatDef} that determine $\eta\vert_p$ in coordinates
 $\{x^i\}_{i=0}^3$ when $\eta\in \Omega^2_2(N)$.
 \end{definition} 

 Let us make three remarks regarding Definition \ref{def:simNotation}.
 First, 
 for the metrics in conditions \ref{def:Sim:B} and \ref{def:Sim:C} we 
 have $\ast_g=\ast_g^{-1}$, so the operations in 
 \ref{def:Sim:B} and \ref{def:Sim:C} are just conjugation by a Hodge star 
 operator. Proposition \ref{prop:lorentzMixing} shows that this is a 
 natural operator in the sense that in the class of pseudo-Riemann metrics, 
 the operation is closed (up to a sign depending on signature).
 Second, if $g=\operatorname{diag}(1,1,1,1)$ and if we use the
 correspondence in Proposition \ref{prop:correspondence}, then
 conjugation $\kappa\mapsto \ast_g^{-1} \circ \kappa \circ \ast_g$ for
 invertible skewon-free $2\choose 2$-tensors corresponds to conjugation
 $G\mapsto \Sigma^t\cdot G\cdot \Sigma$ for area metrics in
 \cite[Theorem 4.10]{Schuller:2010}.
 Second, 
 Proposition \ref{simChar} in Appendix \ref{app:22tensor} gives two  
 alternative descriptions for the property $\kappa\vert_p \sim V$. 
 Third, conditions \ref{def:Sim:A}, \ref{def:Sim:B} and
 \ref{def:Sim:C} are not mutually exclusive.  
 If $\kappa =
 \operatorname{Id}$ then all conditions are equivalent.

 \newcommand{\rr}[0]{\alpha}
 \newcommand{\pos}[0]{\beta}
 \newcommand{\sqr}[0]{\sqrt{2}}

 \begin{theorem}
 \label{thm:classification} 
 Let $\kappa\in \Omega^2_2(N)$, and suppose that
 $\kappaII\vert_p=0$ for some $p\in N$. Then $\kappa\vert_p \sim V$
 for a matrix $V\in \setR^{6\times 6}$ from the below list of matrices
 (listed with Segre type).
 Moreover, if  $\kappa\vert_p$ has Segre type $I, II,\ldots, VI, VII$, then 
 we may assume that $\kappa\vert_p \sim V$ holds with alternative \ref{def:Sim:A} in 
 Definition \ref{def:simNotation}.
 \begin{itemize}
 \item Metaclass I: $[1\overline 1\, 1\overline 1\, 1\overline 1]$
 \begin{eqnarray*}
 \begin{pmatrix}
 \rr_1 & 0 & 0 & -\pos_1 & 0 &0 \\ 
 0 & \rr_2 & 0 & 0  &-\pos_2 & 0 \\
 0 & 0 &  \rr_3 &0 & 0 &-\pos_3  \\
 \pos_1 & 0 &  0 & \rr_1 & 0 &0  \\
 0 & \pos_2 & 0 & 0  &\rr_2 & 0 \\
 0 & 0 & \pos_3 & 0 & 0 &\rr_3 
 \end{pmatrix}
 \end{eqnarray*}

 \item Metaclass II: $[2\overline 2\, 1\overline 1]$
 \begin{eqnarray*}
 \begin{pmatrix}
 \rr_1 &       -\pos_1       &    0 & 0 & 0 & 0 \\
 \pos_1      &    \rr_1    &    0&  0  & 0 & 0 \\
 0              &                0 &  \rr_2 & 0 &0 &-\pos_2 \\
 0              &                1 & 0 & \rr_1 & \pos_1 &0  \\
 1              &                0 & 0 & -\pos_1  &\rr_1 & 0 \\
 0              &                0 & \pos_2 & 0 & 0 &\rr_2 
 \end{pmatrix}
 \end{eqnarray*}
 \item Metaclass III: $[3\overline 3]$
 \begin{eqnarray*}
 \begin{pmatrix}
 \rr_1         &       -\pos_1       &    0 & 0 & 0 & 0 \\
 \pos_1      &    \rr_1           &    0&  0  & 0 & 0 \\
 1              &                0    &  \rr_1 & 0 &0 &-\pos_1 \\
 0              &                0    & 0 & \rr_1 & \pos_1 &1  \\
 0              &                0    & 1 & -\pos_1  &\rr_1 & 0 \\
 0              &                1    & \pos_1 & 0 & 0 &\rr_1
 \end{pmatrix}
 \end{eqnarray*}
 \item Metaclass IV:  $[11\, 1\overline 1\, 1\overline 1]$
 \begin{eqnarray*}
 \begin{pmatrix}
 \rr_1 & 0 & 0 & -\pos_1 & 0 &0 \\
 0 & \rr_2 & 0 & 0  &-\pos_2 & 0 \\
 0 & 0 &  \rr_3 & 0 & 0 &\rr_4  \\
 \pos_1 & 0 &  0 & \rr_1 & 0 &0  \\
 0 & \pos_2 & 0 & 0  &\rr_2 & 0 \\
 0 & 0 & \rr_4 & 0 & 0 &\rr_3
 \end{pmatrix}
 \end{eqnarray*}

 \item Metaclass V:  $[11\, 2\overline 2]$
 \begin{eqnarray*}
 \begin{pmatrix}
 \rr_1 &       -\pos_1           &           0 & 0 & 0 & 0 \\
 \pos_1      &    \rr_1           &           0&  0  & 0 & 0 \\
 0              &                0     &     \rr_2 & 0 &0 &\alpha_3 \\
 0              &                1    &           0 & \rr_1 & \pos_1 &0  \\
 1              &                0    &          0 & -\pos_1  &\rr_1 & 0 \\
 0              &                0    &  \alpha_3 & 0 & 0 &\rr_2 
 \end{pmatrix}
 \end{eqnarray*}
 \item Metaclass VI: $[11\,11\,1\overline 1]$
 \begin{eqnarray*}
 \begin{pmatrix}
 \rr_1 & 0 & 0 & -\pos_1 & 0 &0 \\
 0 & \rr_2 & 0 & 0  &\rr_4 & 0 \\
 0 & 0 &  \rr_3 & 0 & 0 &\rr_5  \\
 \pos_1 & 0 &  0 & \rr_1 & 0 &0  \\
 0 & \rr_4 & 0 & 0  &\rr_2 & 0 \\
 0 & 0 & \rr_5 & 0 & 0 &\rr_3 
 \end{pmatrix}
 \end{eqnarray*}

 \item Metaclass VII: $[1 1\, 1 1\, 1 1]$
 \begin{eqnarray*}
 \begin{pmatrix}
 \rr_1 & 0 & 0 & \rr_4 & 0 &0 \\
 0 & \rr_2 & 0 & 0  &\rr_5 & 0 \\
 0 & 0 &  \rr_3 & 0 & 0 &\rr_6  \\
 \rr_4 & 0 &  0 & \rr_1 & 0 &0  \\
 0 & \rr_5 & 0 & 0  &\rr_2 & 0 \\
 0 & 0 & \rr_6 & 0 & 0 &\rr_3 
 \end{pmatrix}
 \end{eqnarray*}
 \item Metaclass VIII: $[6]$ $\quad \epsilon_1\in \{\pm 1\}$
 \begin{eqnarray*}
 \begin{pmatrix}
 \rr_1 & 0 & 0 & 0 & 0 & 0 \\
 1 & \rr_1 & 0 & 0  & 0 & 0 \\
 0 & 1 &  \rr_1 & 0 &0 &0 \\
 0 & 0 & 0 & \rr_1 & 1 &0  \\
 0 & 0 & 0 & 0  &\rr_1 & 1 \\
 0 & 0 & \epsilon_1 & 0 & 0 &\rr_1 
 \end{pmatrix}
 \end{eqnarray*}
 \item Metaclass IX: $[42]$ $\quad \epsilon_1, \epsilon_2\in \{\pm 1\}$
 \begin{eqnarray*}
 \begin{pmatrix}
 \rr_1 & 0 & 0 & 0 & 0 &0\\
 1 & \rr_1 & 0 & 0  & 0 & 0 \\
 0 & 0 &  \rr_2 &  0 &  0 &0  \\
 0 & 0 & 0 & \rr_1 & 1 &0  \\
 0 & \epsilon_1 & 0 & 0  &\rr_1 & 0 \\
 0 & 0 & \epsilon_2 & 0 & 0 &\rr_2 
 \end{pmatrix}
 \end{eqnarray*}
 \item  Metaclass X: $[4\, 1\overline 1]$ $\quad \epsilon_1\in \{\pm 1\}$
 \begin{eqnarray*}
 \begin{pmatrix}
 \rr_2 & 0 & 0 & 0 & 0 & 0\\
 1 & \rr_2 & 0 & 0  & 0 & 0 \\
 0 & 0 &  \rr_1 & 0 &0 &-\pos_1 \\
 0 & 0 & 0 & \rr_2 & 1 &0  \\
 0 & \epsilon_1 & 0 & 0  &\rr_2 & 0 \\
 0 & 0 & \pos_1 & 0 & 0 &\rr_1 
 \end{pmatrix}
 \end{eqnarray*}
 \item Metaclass XI: $[4 1 1]$ $\quad \epsilon_1\in \{\pm 1\}$
 \begin{eqnarray*}
 \begin{pmatrix}
 \rr_1 & 0      & 0 & 0 & 0 & 0 \\
 1     & \rr_1 & 0&  0  & 0 & 0 \\
 0              & 0       &  \rr_2 & 0 &0 &\rr_3 \\
 0              &                0 & 0 & \rr_1 & 1 &0  \\
 0              &                \epsilon_1 & 0 & 0  &\rr_1 & 0 \\
 0              &                0 & \rr_3 & 0 & 0&\rr_2
 \end{pmatrix}
 \end{eqnarray*}
 \item Metaclass XII: $[2\, 2\overline 2]$  $\quad \epsilon_1\in \{\pm 1\}$
 \begin{eqnarray*}
 \begin{pmatrix}
 \rr_1 &       -\pos_1       &    0 & 0 & 0 & 0 \\
 \pos_1      &    \rr_1    &    0&  0  & 0 & 0 \\
 0              &                0 &  \rr_2 & 0 &0 & 0 \\
 0              &                1 & 0 & \rr_1 & \pos_1 &0  \\
 1              &                0 & 0 & -\pos_1  &\rr_1 & 0 \\
 0              &                0 & \epsilon_1 & 0 & 0 &\rr_2 
 \end{pmatrix}
 \end{eqnarray*}
 \item Metaclass XIII: $[22 2]$  $\quad \epsilon_1,\epsilon_2, \epsilon_3\in \{\pm 1\}, \epsilon_1\le \epsilon_2\le \epsilon_3$
 \begin{eqnarray*}
 \begin{pmatrix}
 \rr_1 & 0 & 0 & 0 & 0 &0 \\
 0 & \rr_2 & 0 & 0  &\epsilon_2 & 0 \\
 0 & 0 &  \rr_3 & 0 & 0 & 0 \\
 \epsilon_1 & 0 &  0 & \rr_1 & 0 &0  \\
 0 & 0 & 0 & 0  &\rr_2 & 0 \\
 0 & 0 & \epsilon_3 & 0 & 0 &\rr_3 
 \end{pmatrix}
 \end{eqnarray*}
 \item Metaclass XIV: $[22\, 1\overline 1]$ $\quad \epsilon_1,\epsilon_2 \in \{\pm 1\}, \epsilon_1\le \epsilon_2$
 \begin{eqnarray*}
 \begin{pmatrix}
 \rr_3 & 0 & 0 & 0 & 0 &0 \\
 0 & \rr_2 & 0 & 0  &\epsilon_2 & 0 \\
 0 & 0 &  \rr_1 & 0 & 0 &-\pos_1 \\
 \epsilon_1 & 0 &  0 & \rr_3 & 0 &0  \\
 0 & 0 & 0 & 0  &\rr_2 & 0 \\
 0 & 0 & \pos_1 & 0 & 0 &\rr_1 
 \end{pmatrix}
 \end{eqnarray*}
 \item Metaclass XV: $[22\, 1 1]$ $\quad \epsilon_1,\epsilon_2 \in \{\pm 1\}, \epsilon_1\le \epsilon_2$
 \begin{eqnarray*}
 \begin{pmatrix}
 \rr_1 & 0 & 0 & 0 & 0 &0 \\
 0 & \rr_2 & 0 & 0  &\epsilon_1 & 0 \\
 0 & 0 &  \rr_3 & 0 & 0 &\rr_4 \\
 \epsilon_2 & 0 &  0 & \rr_1 & 0 &0  \\
 0 & 0 & 0 & 0  &\rr_2 & 0 \\
 0 & 0 & \rr_4 & 0 & 0 &\rr_3 
 \end{pmatrix}
 \end{eqnarray*}
 \item Metaclass XVI: $[2\, 1\overline 1\, 1 \overline1]$  $\quad \epsilon_1\in \{\pm 1\}$
 \begin{eqnarray*}
 \begin{pmatrix}
 \rr_3 & 0 & 0 & 0 & 0 &0 \\
 0 & \rr_1 & 0 & 0  &-\pos_1 & 0 \\
 0 & 0 &  \rr_2 & 0 & 0 &-\pos_2 \\
 \epsilon_1 & 0 &  0 & \rr_3 & 0 &0  \\
 0 & \pos_1 & 0 & 0  &\rr_1 & 0 \\
 0 & 0 & \pos_2 & 0 & 0 &\rr_2 
 \end{pmatrix}
 \end{eqnarray*}
 \item Metaclass XVII: $[21 1\, 1 \overline1]$ $\quad \epsilon_1\in \{\pm 1\}$
 \begin{eqnarray*}
 \begin{pmatrix}
 \rr_2 & 0 & 0 & 0 & 0 &0 \\
 0 & \rr_1 & 0 & 0  &-\pos_1 & 0 \\
 0 & 0 &  \rr_3 & 0 & 0 &\rr_4 \\
 \epsilon_1 & 0 &  0 & \rr_2 & 0 &0  \\
 0 & \pos_1 & 0 & 0  &\rr_1 & 0 \\
 0 & 0 & \rr_4 & 0 & 0 &\rr_3 
 \end{pmatrix}
 \end{eqnarray*}
 \item Metaclass XVIII: $[21 1 1 1]$ $\quad \epsilon_1\in \{\pm 1\}$
 \begin{eqnarray*}
 \begin{pmatrix}
 \rr_1 & 0 & 0 & 0 & 0 &0 \\
 0 & \rr_2 & 0 & 0  &\rr_4 & 0 \\
 0 & 0 &  \rr_3 & 0 & 0 &\rr_5 \\
 \epsilon_1 & 0 &  0 & \rr_1 & 0 &0  \\
 0 & \rr_4 & 0 & 0  &\rr_2 & 0 \\
 0 & 0 & \rr_5 & 0 & 0 &\rr_3 
 \end{pmatrix}
 \end{eqnarray*}
 \item Metaclass XIX:   $[51]$ $\quad \epsilon_1\in \{\pm 1\}$
 \begin{eqnarray*}
 \begin{pmatrix} 
   \alpha_1 & 0 & 0 & 0 & 0 & 0 \\
 1 & \alpha_1 & 0 & 0 & 0 & 0 \\
 0 & \frac{1}{\sqrt{2}} & \frac 1 2 (\alpha_1+\alpha_2) & 0 & 0 & \frac{\epsilon_1}{2} (\alpha_1-\alpha_2) \\
 0 & 0 & 0 &\alpha_1 & 1 & 0 \\
 0 & 0 & \frac{\epsilon_1}{\sqrt{2}} &0 & \alpha_1 & \frac{1}{\sqrt{2}}\\
 0 & \frac{\epsilon_1}{\sqrt{2}} & \frac{\epsilon_1}{2} (\alpha_1-\alpha_2) & 0 & 0 & \frac 1 2 (\alpha_1+\alpha_2)
 \end{pmatrix}
 \end{eqnarray*}
 \item Metaclass XX: $[33]$
 \begin{eqnarray*}
 \begin{pmatrix}
   \alpha_1 & 0 & 0 & 0 & 0 & 0 \\
 \frac{1}{\sqrt{2}} & \frac 1 2 (\rr_1 + \rr_2)
  & -\frac{1}{\sqr} & 0 & \frac 1 2 (\alpha_1 -\alpha_2) & 0 \\
 0 & 0& \alpha_2 & 0 & 0 & 0 \\
 0 & \frac 1 {\sqr} & 0 &\alpha_1 & \frac {1}{\sqr} & 0 \\
 \frac{1}{\sqr} & \frac 1 2( \alpha_1-\alpha_2) & \frac{1}{ \sqrt{2}} &0 & \frac 1 2 (\rr_1 + \rr_2) & 0 \\
 0 & \frac 1 {\sqr} & 0 & 0 & -\frac {1} {\sqr} & \rr_2
 \end{pmatrix}
 \end{eqnarray*}
 \item Metaclass XXI: $[321]$ $\quad \epsilon_1,\epsilon_2\in \{\pm 1\}$
 \begin{eqnarray*}
 \begin{pmatrix} 
   \alpha_1 & 0 & 0 & \epsilon_2 & 0 & 0 \\
 0 & \alpha_2 & \frac{1}{\sqr} & 0 & 0 & \frac{\epsilon_1}{\sqr} \\
 0 & 0 & \frac 1 2 (\alpha_2+\alpha_3) & 0 & \frac{\epsilon_1}{\sqr} & \frac{\epsilon_1}{2} (\alpha_2-\alpha_3) \\
 0 & 0 & 0 &\alpha_1 & 0 & 0 \\
 0 & 0 & 0 &0 & \alpha_2 & 0\\
 0 & 0 & \frac{\epsilon_1}{2} (\alpha_2-\alpha_3) & 0 & \frac{1}{\sqr} & \frac 1 2 (\alpha_2+\alpha_3)
 \end{pmatrix}
 \end{eqnarray*}
 \item Metaclass XXII: $[31\, 1\overline 1 ]$ $\quad \epsilon_1\in \{\pm 1\}$
 \begin{eqnarray*}
 \begin{pmatrix}
   \alpha_1 & 0 & 0 & -\beta_1 & 0 & 0 \\
 0 & \alpha_2 & \frac{1}{\sqr} & 0 & 0 & \frac{\epsilon_1}{\sqr} \\
 0 & 0 & \frac 1 2 (\alpha_2+\alpha_3) & 0 & \frac{\epsilon_1}{\sqr} & \frac{\epsilon_1}{2} (\alpha_2-\alpha_3) \\
 \beta_1 & 0 & 0 &\alpha_1 & 0 & 0 \\
 0 & 0 & 0 &0 & \alpha_2 & 0\\
 0 & 0 & \frac{\epsilon_1}{2} (\alpha_2-\alpha_3) & 0 & \frac{1}{\sqr} & \frac 1 2 (\alpha_2+\alpha_3)
 \end{pmatrix}
 \end{eqnarray*}
 \item Metaclass XXIII: $[31\, 1 1 ]$ $\quad \epsilon_1\in \{\pm 1\}$
 \begin{eqnarray*}
 \begin{pmatrix}
   \alpha_3 & 0 & 0 & \rr_4 & 0 & 0 \\
 0 & \frac 1 2 (\rr_1+\rr_2)  & 0 & 0 & \frac {\epsilon_1}{2}(\rr_1-\rr_2)  & \frac{\epsilon_1}{\sqr} \\
 0 & \frac{1}{\sqr} & \alpha_1 & 0 & \frac{\epsilon_1}{\sqr} & 0 \\
 \rr_4 & 0 & 0 &\alpha_3 & 0 & 0 \\
 0 & \frac {\epsilon_1}{2}(\rr_1-\rr_2)  & 0 &0 & \frac {1}{2}(\rr_1+\rr_2) & \frac{1}{\sqr} \\
 0 & 0 & 0 & 0 & 0 & \rr_1
 \end{pmatrix}
 \end{eqnarray*}
 \end{itemize}
 For each meta-class, $\rr_i\in \setR$, $\pos_i>0$ for $i\in \{1,2,\ldots\}$
 and conditions for signs $\epsilon_i\in \{-1,+1\}$ 
 are given for each metaclass.
 \end{theorem}

 \begin{proof}
 Let $B$ be the $6\times 6$ matrix
 $B=(\varepsilon^{IJ})_{IJ}=H_{(2)}$, where $\varepsilon^{IJ} =
 \varepsilon^{I_1 I_2 J_1 J_2}$ for $I,J\in O$, and 
 $H_{(2)}$ is as in equation \eqref{eq:Hdef2Eq}.

 \textbf{Claim 1.}  For any Segre type $s$ in the list
 \eqref{eq:SegreTypes}, there exists a non-empty finite set of
 invertible $6\times 6$ matrices $\mathscr{S}_s \subset \setR^{6\times
   6}$ with the following property
 \begin{enumerate}
 \item[$(\ast)$]
    If $\kappa\in \Omega^2_2(N)$ is such that $\kappaII\vert_p=0$ and 
    $\kappa\vert_p$ has Segre type $s$, then 
 \begin{eqnarray*}
 \kappa\vert_p &\sim& S\cdot V\cdot S^{-1}.
 \end{eqnarray*}
 for some
  $S\in \mathscr{S}_s$, and a
   Jordan normal form matrix $V\in \setR^{6\times 6}$ with Segre type
  $s$. (See Appendix \ref{app:normalForms} for the definition of Jordan normal form.)
 \end{enumerate}

 To construct $\mathscr{S}_s$,  let $s=\left[m_1\cdots m_r\,
   k_{1}\overline{k_{1}}\cdots k_{s}\overline{k_{s}}\right]$ be a Segre
 type from the list \eqref{eq:SegreTypes}, and 
 let $\mathscr{W}_s\subset \setR^{6\times 6}$ be the
 set of matrices of the form
 \begin{eqnarray*}
   W &=& \bigoplus_{j=1}^r \,\,\epsilon_{j} F_{m_j}\quad\,\,\, \oplus\,\,\, \quad \bigoplus_{j=1}^s F_{2k_j},
 \end{eqnarray*}
 where 
 $\epsilon_1, \ldots, \epsilon_r \in \{\pm 1\}$
 are such that  
 \emph{(i)} $\{\epsilon_j\}_{j=1}^r$ satisfy condition \ref{cond:simuPerB} in Theorem
 \ref{thm:simuDiag} and \emph{(ii)} each $W\in \mathscr{W}_s$ has signature
 $(---+++)$.  It is clear that $\mathscr{W}_s$ is finite and computer algebra shows that
 $\mathscr{W}_s$  is not empty for any $s$.
 If $W\in \mathscr{W}_s$, then $W$ and $B$ are both symmetric matrices
 with spectrum $\{1,1,1,-1,-1,-1\}$ 
 whence there exists an (orthogonal) $S\in \setR^{6\times 6}$ such that
 \begin{eqnarray}
  W &=& \label{eq:WBrel} S^t\cdot B\cdot S.
 \end{eqnarray}  
 Thus, for each $W\in \mathscr{W}_s$ we can find some $S\in
 \setR^{6\times 6}$ such that equation \eqref{eq:WBrel} holds. Let us
 denote one such $S$ by $S=S_W$, and let $\mathscr{S}_s=\{ S_W\in
 \setR^{6\times 6} : W\in \mathscr{W}_s\}$.  
 Let us also note that
 $\mathscr{S}_s$ is not uniquely determined by $s$.

 To show that $\mathscr{S}_s$ satisfies property $(\ast)$,
 let $\kappa\in \Omega^2_2(N)$ be such that
 $\kappaII\vert_p=0$ and $\kappa\vert_p$ has Segre type $s$. 
 Moreover, in coordinates $\{x^i\}_{i=0}^3$ around $p$, let $\kappa^J_I$ be components for $\kappa\vert_p$ 
 as in equation \eqref{eq:kappaMatDef}.
 Then 
 Theorem \ref{theorem:Decomp} implies that
 \begin{eqnarray*}
     \kappa^I_K \varepsilon^{KJ} &=& \kappa^J_K \varepsilon^{KI},\quad I,J\in O.
 \end{eqnarray*}
 For $A=(\kappa_I^J)_{IJ}$ we have $BA=A^t B$, so we can apply Theorem \ref{thm:simuDiag}, and there exists an
 invertible matrix $L\in \setR^{6\times 6}$ such that
 \begin{eqnarray}  
 L^{-1}\cdot A\cdot L &=& \label{eq:A11} V,\\
 L^t\cdot B \cdot L &=&\label{eq:B11} W,
 \end{eqnarray}
 where $V$ is a Jordan normal form matrix with the same Segre type as
 $\kappa\vert_p$ and $W\in \mathscr{W}_s$. Now there exists an $S\in
 \mathscr{S}_s$ such that $W =S^t \cdot B\cdot S$ whence
 \begin{eqnarray}  
 A&=& \label{eq:AA} (SL^{-1})^{-1} \cdot (S\cdot V\cdot S^{-1})\cdot (SL^{-1}),\\
 B&=& \label{eq:BB} (SL^{-1})^t \cdot B\cdot (SL^{-1})
 \end{eqnarray}
 and $\kappa\vert_p\sim S\cdot V\cdot S^{-1}$ follows by Proposition
 \ref{simChar} in Appendix \ref{app:22tensor}.


 If $S\in\setR^{6\times 6}$ is one solution to equation
 \eqref{eq:WBrel}, then the set of all solutions 
  is given by $\{ \Lambda S \in \setR^{6\times 6} : \Lambda^t \cdot
 B\cdot \Lambda = B\}$, and each solution typically gives rise to a
 different normal form for the metaclass. 
 To complete the proof we need to go through all 23 Segre types, and
 for each Segre type $s$, we compute $S\cdot V\cdot S^{-1}$ for all
 $S\in \mathscr{S}_s$ 
 (for a suitable choice of $S$ and $\mathscr{S}_s$)
 and for all
 Jordan normal form matrices $V$ with Segre type $s$.
 The choice of $S$ and $\mathscr{S}_s$ are chosen so that
 normal forms on the theorem formulation correspond to the
 normal forms in  \cite{Schuller:2010} via 
 the correspondence in Proposition \ref{prop:correspondence} with $g=\operatorname{diag}(1,1,1,1)$.

 To show the last claim for Metaclasses I, II, $\ldots$, VI, VII, we
 need to show that the conjugations by Hodge star operators can be
 replaced by coordinate transformations and by possibly redefining the
 constants that appear in the normal form matrices. If $\{x^i\}_{i=0}^3$ are
 coordinates where $\kappa\vert_p \sim V$ holds, let $\{\widetilde
 x^i\}_{i=0}^3$ be coordinates determined by $\widetilde x^i = J^i_j
 x^j$ for a suitable $4\times 4$ matrix $J=(J^i_j)_{ij}$. If
 $g_1=\operatorname{diag}(1,1,1,1)$ and
 $g_2=\operatorname{diag}(1,-1,-1,1)$ are metrics as in Definition
 \ref{def:simNotation} then suitable choices for $J$ are
 $$
 \begin{array}{lccccccccc}
 \mbox{Metaclass}&&&I & II & III & IV & V & VI & VII\\
 \hline
 \mbox{Conjugation by $\ast_{g_1}$} &&
 & J_1 & J_2 & J_3 & J_1 & J_2 & J_1 & \operatorname{Id} \\
 \mbox{Conjugation by $\ast_{g_2}$} &&&J_1 & J_2 & J_2 & J_1 & J_2 & J_1 & \operatorname{Id} 
 \end{array}
 $$
 where $J_1 = \operatorname{diag}(-1,1,1,1)$ and
 \begin{eqnarray*}
 J_2 = 
 \begin{pmatrix}
 0 &0 &0 &1\\
 0 &1 &0 &0\\
 0 &0 &1& 0\\
 1 &0 &0 &0\\
 \end{pmatrix},
 \quad
 J_3 = 
 \begin{pmatrix}
 0& 0& 0& -1\\
 0 &1& 0 &0\\
 0 &0& 1& 0\\
 -1 &0& 0& 0
 \end{pmatrix}.\qedhere
 \end{eqnarray*}
 \end{proof}

 \appendix

 \section{Proposition \ref{simChar}}
 \label{app:22tensor}
 \HOX{Appendix A OK}

 In this appendix we state and prove Proposition \ref{simChar}, which gives
 two alternative descriptions for $\kappa\vert_p\sim V$
 in Definition \ref{def:simNotation}. 

 \begin{proposition}
 \label{simChar}
 Suppose $N$ is a $4$-dimensional manifold,  $\kappa\in \Omega^2_2(N)$, $p\in N$ and
 $V\in \setR^{6\times 6}$.
 Then the following conditions are equivalent
 \begin{enumerate}
 \item \label{thm:condW1}
 $\kappa\vert_p\sim V$.
 \item \label{thm:condW2} 
 There are coordinates $\{x^i\}_{i=0}^3$
 around $p$ and an $\alpha \in \{1,2,3\}$ such that
 \begin{eqnarray}
 \label{eq:simNotationX}
 (\kappa^J_I)_{IJ} &=& H_{(\alpha)} \cdot V\cdot H_{(\alpha)},
 \end{eqnarray}
 where $\kappa^J_I$ are components that represent $\kappa\vert_p$ in
 coordinates $\{x^i\}_{i=0}^3$ as in equation \eqref{eq:kappaMatDef},
 and $H_{(1)}$, $H_{(2)}$, $H_{(3)}$ are the matrices in equations
 \eqref{eq:Hdef1Eq}--\eqref{eq:Hdef2Eq} in Appendix \ref{app:22tensor}.
 \item \label{thm:condW3} There are
   coordinates $\{x^i\}_{i=0}^3$ around $p$ and there
 exists an invertible $S\in \setR^{6\times 6}$ such that
 \begin{eqnarray}
 (  \kappa^J_I)_{IJ} &=& \label{eq:assA1}
  S^{-1}\cdot V\cdot S, \\
   B&=&\label{eq:assA2}
  S^t\cdot B\cdot S,
 \end{eqnarray}
 where $B$ is the $6\times 6$ matrix $B=(\varepsilon^{IJ})_{IJ}=H_{(2)}$.
 \end{enumerate}
 \end{proposition}

 \begin{proof}  
   Equivalence \ref{thm:condW1} $\Leftrightarrow$ \ref{thm:condW2}
   follows since $H_{(2)}$ and $-H_{(3)}$ are matrix representations of
   $\ast_g$ in the basis \eqref{eq:2basis} for metrics
   $g=\operatorname{diag}(1,1,1,1)$ and
   $g=\operatorname{diag}(1,-1,-1,1)$, respectively.
 Implication \ref{thm:condW2}  $\Rightarrow$ \ref{thm:condW3} follows by taking $S=H_{(\alpha)}$.
 For implication \ref{thm:condW3}  $\Rightarrow$ \ref{thm:condW2}, 
 let $A$ be the $6\times 6$ matrix $A=(\kappa^J_I)_{IJ}$, 
  let $\{T^J_I : I,J\in O\}$ be the array of components such that
 $(T^J_I)_{IJ} = S^{-1}$, 
 and for each $J\in O$ let $T^J\in \Lambda^2_{p} (N)$ be defined by
 \begin{eqnarray}
 \label{eq:TIdef}
   T^J &=& T^J_I dx^I.
 \end{eqnarray}
 Equation \eqref{eq:assA2} implies that $ B = \label{eq:BBp} S^{-t}
 \cdot B\cdot S^{-1}$. Thus $\{T^J : J \in O\}$ satisfy the assumptions
 in Proposition \ref{thm:22tensor} whence there exist linearly independent
 covectors $\{\xi^i\}_{i=0}^3$ in $\Lambda_p^1(N)$ such that equation
 \eqref{eq:ThmA2Yeq} holds for some $\alpha \in \{0,1,2,3\}$.  Around
 $p$, let $\{\widetilde x^i\}_{i=0}^3$ be coordinates defined as
 $ \widetilde x^i = \xi^i\left( \left. \pd{}{x^b}\right\vert_p\right) x^b$.
 To see that $\{\widetilde x^i\}_{i=0}^3$ are coordinates it suffices
 to show that $\left( dx^i(u_j)\right)_{ij}$ is the inverse matrix to
 $\left(\pd{\widetilde x^i}{x^j}\right)_{ij}$ when $\{u_i\}_{i=0}^3$ is
 a dual basis to $\{\xi^i\}_{i=0}^3$.  Thus $\xi^i = d\widetilde
 x^i\vert_p$ and equations \eqref{eq:ThmA2Yeq}, \eqref{eq:TIdef}
 and $d\widetilde x^I = \pd{\widetilde x^I}{x^L} dx^L$ imply that
 $T^J_I = Y^J_{(\alpha)K} \pd{\widetilde x^K}{x^I}$.  Equation \eqref{eq:assA1}
 further implies that $A \cdot S^{-1} = S^{-1} \cdot V$ and by equation 
 \eqref{eq:kappaTransRule},
 \begin{eqnarray*}
  \widetilde \kappa^L_I Y^{J}_{(\alpha)L} &=& Y^L_{(\alpha)I} V_{b(L) b(J)}, \quad I,J\in O.
 \end{eqnarray*}
 Since $H_{(\alpha)}^2 = \operatorname{Id}$ 
 it follows that 
 $( \widetilde \kappa^J_I)_{IJ} =H_{(\alpha)}\cdot V\cdot H_{(\alpha)}$
 where $\alpha\in \{0,1,2,3\}$,
 and part \ref{thm:condW2} follows. 
 \end{proof}

 \begin{proposition} 
 Suppose 
 \label{thm:22tensor}
 $ T^{I}\in \Lambda^2_p(N) $ for all $I\in O$ on a $4$-manifold $N$,
 where $O$ is as in Section \ref{sec:Rep6x6} and $p\in N$.
 Moreover, suppose that
 \begin{eqnarray}
 \label{eq:TIJcond}
   T^{I}\wedge T^{J} &=& \varepsilon^{IJ} \omega, \quad I,J\in O
 \end{eqnarray}
 for some $\omega\in \Lambda^4_p(N)\slaz$.
 Then there exists linearly independent $\xi_0,\ldots, \xi_3 \in
 \Lambda^1_p(N)$ and an $\alpha\in \{0,1,2,3\}$ such that
 \begin{eqnarray}
 \label{eq:ThmA2Yeq}
    T^J &=& Y^J_{(\alpha)I} \xi^I, \quad J\in O,
 \end{eqnarray}
 where $\xi^I=\xi^{I_1}\wedge \xi^{I_2}$ and $Y^J_{(\alpha)I}$ are
 components such that $(Y^J_{(\alpha)I})_{IJ}=H_{(\alpha)}$ for one of
 the $6\times 6$ matrices
 \begin{eqnarray}
 \label{eq:Hdef1Eq}
 \quad\quad H_{(0)} &=& -\operatorname{Id},\quad\quad \quad \quad\quad\quad \quad\quad\,\,
 \,\,\,H_{(1)} \,\,=\,\, \operatorname{Id},\\ 
 \quad\quad H_{(2)} &=& 
 \label{eq:Hdef2Eq}
 \begin{pmatrix}
 & &  & 1&  &  \\
 &  & & & 1 &  \\
 &  & & &  &1  \\
 1&  & &&  &  \\
 & 1 & &&  &  \\
 &  &1 &&  &  
 \end{pmatrix}, \quad   
 H_{(3)}\,\, =\,\, \begin{pmatrix}
 & &  & 1&  &  \\
 &  & & & 1 &  \\
 &  & & &  & -1 \\
 1&  & &&  &  \\
 & 1 & &&  &  \\
 &  &-1 &&  &  
 \end{pmatrix}.
 \end{eqnarray}
 \end{proposition}

 \begin{proof} 
   Let us first note that equation \eqref{eq:TIJcond} implies that
   $T^{I}$ is non-zero for each $I\in O$.  Let $g$ be an auxiliary
   positive definite Riemann metric on $N$. Then
 \begin{eqnarray*}
 \langle u,v\rangle &=& \ast_g (u\wedge v), \quad u,v\in \Lambda^2_p(N)
 \end{eqnarray*}
 defines an indefinite inner product in $\Lambda^2_p(N)$ of signature
 $(+++---)$ \cite{Harnett:1991}. For a vector subspace $W\subset
 \Lambda^2_p(N)$, we denote the \emph{orthogonal complement} (with
 respect to $\langle \cdot, \cdot\rangle$) by $W^\perp$ \cite[p. 49]{ONeill:1983}.

 \textbf{Claim 1.} There exists linearly independent covectors $\xi^0,
 \xi^1, \xi^2\in \Lambda^1_p(N)$ such that
 \begin{eqnarray}
 \label{eq:SolveEqClaim1}
 T^{0i} &=& \xi^0\wedge \xi^i, \quad i\in \{1,2\}.
 \end{eqnarray}

 In four dimensions, the Plucker identities states that a $q\in \Lambda^2_p(N)$ can be
 written as $q=a\wedge b$ for some $a,b\in \Lambda^1_p(N)$ if and
 only if $q\wedge q=0$ \cite[p. 184]{Cohn:2005}.
 Thus equation \eqref{eq:TIJcond}
 implies that there exist $\xi_0,\xi_1 \in \Lambda^1_p(N)$ such that
 \begin{eqnarray}
 \label{eq:ClT01}
  T^{01}=\xi^0\wedge \xi^1.
 \end{eqnarray}
 Since $T^{01}\neq 0$, covectors $\xi^{0}$ and $\xi^{1}$ are linearly independent.
 Let $\xi^2,\xi^3\in \Lambda^1_p(N)$ be such that $\{\xi^i\}_{i=0}^3$ is a basis for $\Lambda^1_p(N)$.
 For $W = \operatorname{span}\{ T^{01} \}$ we then have
 $\dim W^\perp = 5$ and
 \begin{eqnarray*}
 W^\perp &=& \operatorname{span}\{ \{\xi^0\wedge \xi^r\}_{r=1}^3,\, \{ \xi^1 \wedge \xi^s \}_{s=2}^3 \}.
 \end{eqnarray*}
 Since $T^{02} \in W^\perp$ we have
 \begin{eqnarray}
 \label{eq:SDef}
  T^{02} &=& A \xi^0\wedge \xi^1 + \xi^0\wedge \zeta^0 + \xi^1\wedge \zeta^1,
 \end{eqnarray}
 for some $A\in \setR$ and $\zeta^0, \zeta^1\in \operatorname{span} \{\xi^i\}_{i=2}^3$.
 From
 $T^{02}\wedge T^{02}=0$, it follows that
 $
  \xi^0\wedge \xi^1\wedge \zeta^0 \wedge \zeta^1=0.
 $
 Thus covectors $\xi^0, \xi^1, \zeta^0, \zeta^1$ are linearly dependent and there are 
 constants $C_i$ such that
 \begin{eqnarray}
 \label{eq:linDep}
 C_1 \xi^0  + C_2 \xi^1 +   C_3 \zeta^0  + C_4 \zeta^1 &=&0,
 \end{eqnarray}
 and all $C_1, C_2, C_3, C_4$ are not zero. Since $C_1=C_2=0$, we can not have $C_3 = C_4=0$.
 If $C_3\neq 0$, then $\zeta^0 = -\frac{C_4}{C_3} \zeta^1$ and equations \eqref{eq:ClT01} and \eqref{eq:SDef} yield
 $$
 T^{01} = \left( \xi^1- \frac {C_4}{C_3} \xi^0\right) \wedge (-\xi^0),\quad
 T^{02} = \left( \xi^1- \frac {C_4}{C_3} \xi^0\right) \wedge \left(\zeta^1 - A\xi^0\right).
 $$
 If $\zeta^1=0$ then $\zeta^0=0$ whence equation \eqref{eq:SDef} implies that $T^{02}=AT^{01}$
 and $A\neq 0$.
 Writing out $AT^{01}\wedge T^{23} = T^{02}\wedge T^{23}$ using equation \eqref{eq:TIJcond}
 gives a contradiction, so $\zeta^1\neq 0$. Hence 
 covectors $\xi^1- \frac {C_4}{C_3} \xi^0, -\xi^0$ and $\zeta^1 - A\xi^0$ are linearly independent
 and Claim 1 follows.
 The case $C_4\neq 0$ follows similarly.

 \textbf{Claim 2.} 
 There exists a basis $\{\xi^i\}_{i=0}^3$ for $\Lambda^1_p(N)$ such that
 equations \eqref{eq:SolveEqClaim1} hold and
 \begin{eqnarray}
 \label{eq:claim2}
  T^{03} &=& \xi^0 \wedge \zeta + D \xi^1\wedge \xi^2
 \end{eqnarray}
 for some $D\in \setR$ and $\zeta \in \operatorname{span}\{\xi^i\}_{i=1}^3$.

 If $\xi^0, \xi^1, \xi^2 \in \Lambda^1_p(N)$ are as in Claim 1, then there exists a
 $\xi^3 \in \Lambda^1_p(N)$ 
 such that
 $\{\xi^i\}_{i=0}^3$ is a basis for $\Lambda^1_p(N)$.  For
 $W=\operatorname{span}\{T^{01}, T^{02} \}$, we then have $\dim W^\perp = 4$ 
 and
 \begin{eqnarray*}
  W^\perp &=& \operatorname{span} \{ \xi^1\wedge \xi^2, \,\,  \{
 \xi^{0}\wedge \xi^i \}_{i=1}^3\}.
 \end{eqnarray*}
Claim 2 follows since $T^{03}\in W^\perp$ 

In Claim 2 we may assume that $D$ and $\zeta$ are not both zero since $T^{03}\neq 0$.
The proof then divides into three cases: 
$D=0, \zeta \neq 0$ (Claim 3),
$D\neq 0, \zeta \neq 0$ (Claim 4)
and
$D\neq 0, \zeta = 0$ (Claim 5).

 \textbf{Claim 3.} If Claim 2 holds with $D=0$ and $\zeta\neq 0$, then there are linearly independent 
 $\xi^0, \ldots, \xi^3 \in \Lambda^1_p(N)$ and a $\tau \in \{\pm 1\}$
 such that
 \begin{eqnarray}
 T^{0i} &=& 
 \label{cl3:eqA}
 \xi^0 \wedge \xi^i, \quad i\in \{1,2,3\}, \\
 T^{12} &=& \label{cl3:eqB} \tau \xi^1 \wedge \xi^2, \\
 T^{23} &=&\label{cl3:eqC} \tau \xi^2 \wedge \xi^3,\\
 T^{31} &=& \label{cl3:eqD}\tau \xi^3 \wedge \xi^1.
 \end{eqnarray}

 The proof is divided into four steps.  In Step $1$, let us show that
 there are linearly independent $\{\xi^i\}_{i=0}^3$ such that equations
 \eqref{cl3:eqA} hold.  Since $D=0$ in Claim 2, equations
 \eqref{cl3:eqA} hold by setting $\xi^3 = \zeta$. Therefore we only need to
 show that $\{\xi^i\}_{i=0}^3$ are linearly independent.  If there are
 constants $C_0,\ldots, C_3 \in \setR$ such that $\sum_{i=0}^3 C_i \xi^i = 0$,
 then
 \begin{eqnarray*}
 C_1 T^{01} + C_2 T^{02} + C_3 T^{03}  &=&0.
 \end{eqnarray*}
 Thus $C_1 T^{01}\wedge T^{23}=0$ so $C_1=0$ by equation \eqref{eq:TIJcond}. 
 Similarly we obtain $C_2=C_3=0$. Thus $C_0 \xi^0=0$, so $C_0=0$, and
 $\{\xi^i\}_{i=0}^3$ are linearly independent.

 In Step 2, let us show that there exists a $\tau\in \{\pm 1\}$ and  linearly independent
 $\{\xi^i\}_{i=0}^3$ such that equations \eqref{cl3:eqA}--\eqref{cl3:eqB} hold.
 By Step 1, there are linearly independent $\{\xi^i\}_{i=0}^3$ such that
 equations \eqref{cl3:eqA} hold.
 We know that $T^{12} \in \operatorname{span}\{T^{01}, T^{02}\}^\perp$. Hence
 \begin{eqnarray}
 \label{eq:T12genExp}
 T^{12} &=& \xi^0 \wedge \zeta + E \xi^1\wedge \xi^2
 \end{eqnarray}
 where $\zeta\in \operatorname{span}\{\xi^i\}_{i=1}^3$ and $E \in
 \setR$. Equations \eqref{cl3:eqA}, \eqref{eq:TIJcond} and
 \eqref{eq:T12genExp} imply that $\omega=T^{03}\wedge T^{12}=E \xi^0\wedge
 \xi^1\wedge \xi^2\wedge \xi^3$, so $E\neq 0$. Let $\tau =
 \operatorname{sgn} E$.  Since $T^{12}\wedge T^{12}=0$, it follows that
 $\xi^0, \xi^1, \xi^2,\zeta$ are linearly dependent and there are
 constants $C_0, \ldots, C_3$ such that
 $$
  C_0 \xi^0 + C_1 \xi^1 + C_2 \xi^2 + C_3\zeta=0
 $$
 and all $C_0, \ldots, C_3$ are not zero. It is clear that
 $C_0=0$. Since $C_3=0$ is not possible, there are constants $A,B\in
 \setR$ such that
 \begin{eqnarray*}
 T^{12} &=&  \xi^0 \wedge \left( A\xi^1 + B\xi^2\right) + E\xi^1\wedge \xi^2.
 \end{eqnarray*}
 Thus
 \begin{eqnarray*}
 T^{01} &=& \left( \frac 1 {\sqrt{\vert E\vert }} \xi^0\right) 
 \wedge \left( \sqrt{\vert E\vert} \xi^1 + \frac {\tau B} {\sqrt{\vert E\vert}} \xi^0\right), \\
 T^{02} &=&  \left( \frac 1 {\sqrt{\vert E\vert }} \xi^0\right) 
 \wedge \left( \sqrt{\vert E\vert} \xi^2 - \frac {\tau A} {\sqrt{\vert E\vert}} \xi^0\right), \\
 T^{03} &=&  \left( \frac 1 {\sqrt{\vert E\vert }} \xi^0\right) 
 \wedge \left( \sqrt{\vert E\vert} \xi^3\right),\\
 T^{12} &=& \tau \left(  \sqrt{\vert E\vert} \xi^1 + \frac {\tau B}  {\sqrt{\vert E\vert}} \xi^0\right) \wedge \left( \sqrt{\vert E\vert} \xi^2 - \frac{\tau A} {\sqrt{\vert E\vert}} \xi^0\right).
 \end{eqnarray*}
 Since the four covectors inside the parentheses are linearly
 independent, Step 2 follows.

 In Step 3, let us show that there exists a $\tau \in \{\pm 1\}$ and
 linearly independent $\{\xi^i\}_{i=0}^3$ such that equations
 \eqref{cl3:eqA}--\eqref{cl3:eqC} hold. By Step 2, there exists a
 $\tau\in \{\pm 1\}$ and linearly independent covectors
 $\{\xi^i\}_{i=0}^3$ such that \eqref{cl3:eqA}--\eqref{cl3:eqB} hold.
 Since $T^{23} \in \operatorname{span}\{T^{12}, T^{02}, T^{03}\}^\perp$
 we have
 \begin{eqnarray*}
 T^{23} &=& \xi^0 \wedge \zeta + F \xi^2\wedge \xi^3
 \end{eqnarray*}
 for some $\zeta\in \operatorname{span} \{\xi^a\}_{a=1}^2$ and $F\in \setR$.
 Writing out $T^{01}\wedge T^{23}= T^{03}\wedge T^{12}$ shows that $F=\tau$.
 Since $T^{23}\wedge T^{23}=0$ there are
 constants $C_0, \ldots, C_3$ such that
 \begin{eqnarray*}
 C_0 \xi^0 + C_1 \xi^2 + C_2 \xi^3 + C_3\zeta&=&0,
 \end{eqnarray*}
 and all $C_0, \ldots, C_3$ are not zero.  Now $C_0=0$ and
 $C_2=0$. Since $C_3\neq 0$ is not possible, it follows that $\zeta = C
 \xi^2$ for some $C\in \setR$.  Thus $T^{23} = \tau \xi^2\wedge
 (\xi^3-\tau C\xi^0)$, and Step 3 follows by rewriting $T^{01}, T^{02},
 T^{03}, T^{12}, T^{23}$ and checking linear independence as in Step 2.

 In Step 4, let us show that there exists a $\tau\in \{\pm1\}$ and
 linearly independent $\{\xi^i\}_{i=0}^3$ such that equations
 \eqref{cl3:eqA}--\eqref{cl3:eqD} hold. By Step 3, there exist a
 $\tau\in \{\pm 1\}$ and linearly independent covectors
 $\{\xi^i\}_{i=0}^3$ such that equations \eqref{cl3:eqA}--\eqref{cl3:eqC} hold.
 Since $T^{31} \in
 \operatorname{span}\{T^{12}, T^{01}, T^{03}, T^{23}\}^\perp$ it follows that
 \begin{eqnarray*}
     T^{31} &=& A \xi^3 \wedge \xi^1 + B \xi^0 \wedge \xi^2
 \end{eqnarray*}
 for some $A,B\in \setR$. Writing out $T^{02}\wedge T^{31} =
 T^{01}\wedge T^{23}$ gives $A=\tau$ and writing out $T^{31}\wedge
 T^{31} = 0$ gives $B=0$.  This completes the proof of Claim 3.

\textbf{Claim 4.} Suppose Claim 2 holds with $D\neq 0$, $\zeta\neq 0$. If $\sigma
= \operatorname{sgn} D$, then there are linearly independent $\xi^0, \ldots, \xi^3
\in \Lambda^1_p(N)$ such that
\begin{eqnarray}
  T^{01} &=& \label{cl4:eqA} \xi^2 \wedge \xi^3, \\
  T^{02} &=& \label{cl4:eqB}  \xi^3 \wedge \xi^1, \\
  T^{03} &=& \label{cl4:eqC} \sigma \xi^1 \wedge \xi^2, \\
  T^{12} &=& \label{cl4:eqD} \sigma \xi^0 \wedge \xi^3, \\
  T^{23} &=& \label{cl4:eqE} \xi^0 \wedge \xi^1,\\
  T^{31} &=& \label{cl4:eqF} \xi^0 \wedge \xi^2.
\end{eqnarray}

As the proof of Claim 3,  the proof is divided into four
steps.  In Step 1, let us show that there are linearly independent
$\{\xi^i\}_{i=1}^3$ such that equations
\eqref{cl4:eqA}--\eqref{cl4:eqC} hold.  
Let $\{\xi^i\}_{i=0}^3$, $D\neq 0$ and  $\zeta\neq 0$ be as in Claim 2.
Then 
$T^{03}\wedge T^{03}=0$ implies that
$\zeta=A\xi^1 + B\xi^2$ for some $A,B\in \setR$, and
\begin{eqnarray*}
T^{01} &=& \left( -\sqrt{\vert D\vert} \xi^1 - \frac {\sigma B} {\sqrt{ \vert D\vert}} \xi^0\right) \wedge \left( \frac 1 {\sqrt{ \vert D\vert}} \xi^0\right), \\
T^{02} &=& \left( \frac 1 {\sqrt{\vert D\vert }} \xi^0\right) \wedge \left( \sqrt{\vert D\vert} \xi^2 - \frac {\sigma A} {\sqrt{\vert D\vert}} \xi^0\right), \\
T^{03} &=& \sigma \left(  \sqrt{\vert D\vert} \xi^2 - \frac {\sigma A}  {\sqrt{\vert D\vert}} \xi^0\right) \wedge \left( -\sqrt{\vert D\vert} \xi^1 - \frac{\sigma B} {\sqrt{\vert D\vert}} \xi^0\right).
\end{eqnarray*}
Step 1 follows since the covectors in the parentheses are linearly
independent.

In Step 2, let us show that there are linearly independent
$\{\xi^i\}_{i=0}^3$ such that equations
\eqref{cl4:eqA}--\eqref{cl4:eqD} hold.  By Step 1, there are linearly
independent $\{\xi^i\}_{i=1}^3$ such that equations
\eqref{cl4:eqA}--\eqref{cl4:eqC} hold. We know that $T^{12} \in
\operatorname{span}\{T^{01}, T^{02}\}^\perp$. Hence
\begin{eqnarray*}
T^{12} &=& \xi^3 \wedge \zeta + E \xi^1\wedge \xi^2
\end{eqnarray*}
where $\zeta\in \Lambda^1_p(N)$ and $E \in \setR$. Writing out
$T^{03}\wedge T^{12}\neq 0$ and $ T^{12}\wedge T^{12} = 0$ shows that
$\xi^1, \xi^2, \xi^3, \zeta$ are linearly independent and $E=0$.  Step
2 follows by setting $\xi^0 = -\sigma \zeta$.

In Step 3, let us show that there are linearly independent
$\{\xi^i\}_{i=0}^3$ such that equations
\eqref{cl4:eqA}--\eqref{cl4:eqE} hold. By Step 2, there exist linearly
independent $\{\xi^i\}_{i=0}^3$ such that
\eqref{cl4:eqA}--\eqref{cl4:eqD} hold.  Since $T^{23} \in
\operatorname{span}\{T^{12}, T^{02}, T^{03}\}^\perp$, it follows that
\begin{eqnarray*}
   T^{23} &=& \xi^1 \wedge (A\xi^0+ B\xi^3) + E \xi^2\wedge \xi^3
\end{eqnarray*}
for some $A,B,E \in \setR$.
Writing out $T^{01}\wedge T^{23}=T^{03}\wedge T^{12}$ 
and $T^{23}\wedge T^{23}=0$ gives 
$A=-1$ and $E=0$.
Thus
\begin{eqnarray*}
T^{23} &=& (\xi^0 - B\xi^3)\wedge \xi^1,
\end{eqnarray*}
and Step 3 follows since $T^{12}$ can be rewritten as 
$T^{12} =  \sigma (\xi^0 - B\xi^3) \wedge \xi^3$.

In Step 4, let us show that there are linearly independent
$\{\xi^i\}_{i=0}^3$ such that equations
\eqref{cl4:eqA}--\eqref{cl4:eqF} hold. By Step 3, there exist linearly
independent covectors $\{\xi^i\}_{i=0}^3$ such that
\eqref{cl4:eqA}--\eqref{cl4:eqE} hold. Since $T^{31} \in
\operatorname{span}\{T^{01}, T^{03}, T^{12}, T^{23}\}^\perp$ we have
\begin{eqnarray*}
T^{31} &=& A \xi^3 \wedge \xi^1 + B \xi^0 \wedge \xi^2
\end{eqnarray*}
for some $A,B\in \setR$.  Writing out $T^{31}\wedge
T^{02}=T^{01}\wedge T^{23}$ and $T^{31}\wedge T^{31}=0$  gives
$B=1$ and $A=0$, so
equation \eqref{cl4:eqF} holds and Step 4 follows. This completes the
proof of Claim 4.

 \textbf{Claim 5.} Suppose Claim 2 holds with $D\neq 0$ and $\zeta=0$ and let $\sigma
= \operatorname{sgn} D$. Then there
 are linearly independent $\xi^0, \ldots, \xi^3 \in \Lambda^1_p(N)$ such that
\begin{eqnarray}
 T^{0i} &=& \label{cl5:eqA} \xi^0 \wedge \xi^i, \quad i\in \{1,2\}, \\
 T^{03} &=& \label{cl5:eqB}  \sigma \xi^1 \wedge \xi^2, \\
T^{12} &=& \label{cl5:eqD}   \xi^0 \wedge \xi^3,\\
T^{23} &=& \label{cl5:eqE}   \sigma \xi^2 \wedge \xi^3, \\
T^{31} &=& \label{cl5:eqF}   \sigma \xi^3 \wedge \xi^1.
\end{eqnarray}

The proof of Claim 5 is divided into three steps.  In Step 1, we show
that there are linearly independent $\{\xi^i\}_{i=0}^3$ such that
equations \eqref{cl5:eqA}--\eqref{cl5:eqD} hold.  The argument for
Claim 2 shows that
\begin{eqnarray*}
T^{12} &=&   \xi^0\wedge \eta + E \xi^1 \wedge \xi^2
\end{eqnarray*}
for some $E\in \setR$ and $\eta\in \operatorname{span}\{
\xi^i\}_{i=0}^3$.  Writing out $T^{12}\wedge T^{03}\neq 0$ shows that
$\{\xi^0, \xi^1, \xi^2, \eta\}$ are linearly independent.  Then
$T^{12}\wedge T^{12}=0$ implies that $E=0$ and equations
\eqref{cl5:eqA}--\eqref{cl5:eqD} follow by setting $\xi^3 = \sqrt{\vert D\vert} \eta$ 
and suitably scaling $\xi^0, \xi^1, \xi^2$.

In Step 2, we show that there are linearly independent
$\{\xi^i\}_{i=0}^3$ such that equations
\eqref{cl5:eqA}--\eqref{cl5:eqE} hold.  By the argument in Claim 3,
Step 3, there is a $B\in \setR$ such that
\begin{eqnarray*}
T^{23} &=&  \sigma  \xi^2\wedge \left( \xi^3 +  \sigma B  \xi^0\right)
\end{eqnarray*}
and equations \eqref{cl5:eqA}--\eqref{cl5:eqD} follow by redefining $\xi^3 \mapsto \xi^3 + \sigma B  \xi^0$.

In Step 3, we show that there are linearly independent
$\{\xi^i\}_{i=0}^3$ such that equations
\eqref{cl5:eqA}--\eqref{cl5:eqF} hold. This follows by repeating the
argument in Claim 3, Step 4.

We can now complete the proof.
When Claim 3 holds, then equation \eqref{eq:ThmA2Yeq} follows by replacing covectors 
$\xi^i \mapsto \{\tau \xi^0, \xi^1, \xi^2, \xi^3\}$ and
$\alpha = 0$ when $\tau = -1$ and
$\alpha = 1$ when $\tau = 1$.
 When Claim 4 holds, then equation \eqref{eq:ThmA2Yeq} holds with 
 $\alpha = 2$ when $\sigma = 1$ and
 $\alpha =3$ when $\sigma = -1$.
When Claim 5 holds, then equation \eqref{eq:ThmA2Yeq} follows by replacing covectors 
$\xi^i \mapsto \{\xi^3, -\xi^2, \xi^1, -\sigma \xi^0\}$ and
$\alpha = 2$ when $\sigma = 1$ and
$\alpha = 3$ when $\sigma = -1$.
 \end{proof}

 \section{Normal form for a $H$-selfadjoint matrix}
 \HOX{Appendix B OK}
 \label{app:normalForms}
 The Jordan normal form theorem (Theorem \ref{thm:jordan}) is a
 fundamental theorem in linear algebra. 
 In this appendix we formulate Theorem \ref{thm:simuDiag} which extends
 this result to two matrices that are suitably compatible. The result
 is known as the \emph{canonical form of an $H$-selfadjoint
  matrix}. The result and its proof can be found in \cite[Theorem
 12.2]{LancasterRodman:2005}.

 First we define the block matrices that appear in the Jordan normal form
 theorem for real matrices \cite[Theorem 2.2]{LancasterRodman:2005}.
 For $m\in \{1,2,\ldots\}$, $\lambda, \sigma \in \setR$ and $\tau>0$ let
 \begin{eqnarray*}
 R_{m}(\lambda) &=& 
 \begin{pmatrix}
 \lambda  &  1           &           &                    &  \\
               & \lambda &  1       &                   &   \\
               &               &\ddots &   \ddots      &   \\
               &               &            & \lambda    & 1\\
              &               &             &                  &  \lambda 
 \end{pmatrix}\in \setR^{m\times m},\\
 C_{2m}(\sigma\pm i\tau) &=& 
 \begin{pmatrix}
  \sigma & \tau & 1 & 0 &  &  &  &&\\
  -\tau & \sigma & 0 & 1 &  &  &&  &\\
         & &   \sigma & \tau & 1 &  0&&&  \\
         & &   -\tau & \sigma & 0 & 1 &&&  \\
 &&&&\ddots &&\ddots& \\
 &&&&  &\ddots& &   1 & 0 \\
 &&&& && \ddots&   0 & 1 \\
 & &&       &  && &   \sigma &\tau \\
 & &&      &  && &   -\tau & \sigma  
 \end{pmatrix}\in \setR^{2m\times 2m}.
 \end{eqnarray*}

 Moreover, let $F_1=(1)$ and for $m\ge 2$, let $F_m$ be the
 \emph{standard involutary permutation} matrix
 \begin{eqnarray}
 \label{eq:Fmdef}
 F_{m} &=& 
 \begin{pmatrix}
                &    &                           1 \\
               &     
 \mathrel{\raisebox{-1.0em}{
      \reflectbox{\rotatebox[origin=c]{180}{$\ddots$}}
 } }&          \\
 1              &               &                        
 \end{pmatrix}\in \setR^{m\times m}.
 \end{eqnarray}

 For square matrices $M_1, \ldots, M_k$, we define
 \begin{eqnarray}
 \label{eq:oPlusBlockDef}
 M_1\oplus \cdots \oplus M_k &=&
 \begin{pmatrix}
 M_1& &\\
 & \ddots & \\
 &  & M_k\\
 \end{pmatrix}.
 \end{eqnarray}

 The next theorem is the \emph{Jordan normal form theorem} with the ordering
 in equation \eqref{eq:jordanInEq} being a consequence of Proposition
 \ref{prop:jordanPermutationBlocks}.  We say that a matrix $A\in
 \setR^{n\times n}$ is in \emph{Jordan normal form} if Theorem \ref{thm:jordan}
 holds with $L=\operatorname{Id}$.

 \begin{theorem}
 \label{thm:jordan}
 Suppose $A\in \setR^{n\times n}$. Then  there exists an invertible matrix $L\in \setR^{n\times n}$ such that
 \begin{eqnarray}
 \label{eq:jordanDecomp}
 L^{-1} A L &=& \bigoplus_{j=1}^r \,\,R_{m_j}(\lambda_j)\,\,\, \oplus\,\,\,\, \bigoplus_{j=1}^s C_{2k_j}(\sigma_j \pm i\tau_j), 
 \end{eqnarray}
 for some $r,s\ge 0$, $\lambda_1, \ldots, \lambda_r\in \setR$, 
 $\sigma_1, \ldots, \sigma_s\in \setR$,
 $\tau_{1}, \ldots, \tau_s>0$ and
 \begin{eqnarray}
 \label{eq:jordanInEq}
 m_1 \ge \cdots\ge m_r \ge 1, \quad k_{1}\ge \cdots \ge k_s\ge 1.
 \end{eqnarray}
 Moreover, suppose that $\widetilde L$ is another $n\times n$ matrix such that
 equations \eqref{eq:jordanDecomp} and \eqref{eq:jordanInEq} hold for
 block matrices 
 $(R_{\widetilde m_j}(\widetilde \lambda_j))_{j=1}^{\widetilde r}$  and
 $(C_{2 \widetilde k_j}(\widetilde \sigma_j\pm i\widetilde
 \tau_j))_{j=1}^{\widetilde s}$. Then $\widetilde r=r$, $\widetilde s = s$
 and 
 $(R_{\widetilde  m_j}(\widetilde \lambda_j))_{j=1}^r$ 
 is a permutation of 
 $(R_{m_j}(\lambda_j))_{j=1}^r$ 
 and 
 $(C_{2\widetilde k_j}(\widetilde \sigma_j\pm i\widetilde \tau_j))_{j=1}^s$
 is a permutation of
 $(C_{2k_j}(\sigma_j\pm i\tau_j))_{j=1}^s$.
 In particular, $\widetilde m_j = m_j$ for $j=1, \ldots, r$ and
 $\widetilde k_j = k_j$ for $j=1, \ldots, s$.
 \end{theorem}

 The next proposition shows that the blocks in $M_1\oplus \cdots \oplus M_k$ can
 be permutated into any order using a similarity transformation
 \cite[p. 31]{Fiedler:1986}.

 \begin{proposition}
 \label{prop:jordanPermutationBlocks}
 Suppose 
 \begin{eqnarray*}
 A &=& M_1\oplus \cdots \oplus M_k,
 \end{eqnarray*}
 where $M_1, \ldots, M_k$ are real square matrices, and suppose that
 $\pi$ is a permutation of $\{1,2,\ldots, k\}$. Then there exists a 
 real orthogonal matrix $P$ such that
 \begin{eqnarray*}
 P^{-1} A P &=& M_{\pi(1)} \oplus \cdots \oplus M_{\pi(k)}.
 \end{eqnarray*}
 \end{proposition}

 For example, if $M_1\in \setR^{n\times n}$ and $M_2\in \setR^{m\times m}$ then 
 $P^{-1}\cdot (M_1\oplus M_2) \cdot P = M_2\oplus M_1$ for 
 $P = 
 \begin{pmatrix}
 0_{n\times m} & I_{n\times n}\\
 I_{m\times m} & 0_{m\times n}
 \end{pmatrix}$,
 where $0_{a\times b}$ is the $a\times b$ zero matrix, and $I_{a\times a}$ is the $a\times a$ identity matrix.

 \comment{
 Let $Q = 
 \begin{pmatrix} 
 0_{n\times m} & I_{n\times n}\\ 
 I_{m\times m} & 0_{m\times n}
 \end{pmatrix}$, 
 where $0_{a\times b}$ is the $a\times b$ zero matrix, and $I_{a\times a}$ is the $a\times a$ identity matrix. 
 Then $Q^{t} = Q^{-1}$ and 
 \begin{eqnarray*} 
 Q^{-1}\cdot \begin{pmatrix} 
 A & 0\\ 
 0 & B
 \end{pmatrix}\cdot Q &=&
 \begin{pmatrix} 
 B & 0\\ 
 0 & A
 \end{pmatrix}, 
 \end{eqnarray*} 
 for all $A\in \setR^{n\times n}$ and $B\in \setR^{m\times m}$. 
 The result follows since any permutation can be written as a composition of transpositions
 of adjecent elements.}

 \begin{theorem}
 \label{thm:simuDiag}
 Suppose $A,B\in \setR^{n\times n}$ are matrices such that 
 $$
 B=B^t, \quad \det B \neq 0, \quad BA=A^tB.
 $$
 Then
 there exists an invertible  $n\times n$ matrix $L$ such that
 \begin{eqnarray*}
 L^{-1} A L &=& \bigoplus_{j=1}^r \,\,R_{m_j}(\lambda_j)\,\,\, \oplus\,\,\,\, \bigoplus_{j=1}^s C_{2k_j}(\sigma_j \pm i\tau_j), \\
 L^t B L &=& \bigoplus_{j=1}^r \,\,\epsilon_{j} F_{m_j}\quad\,\,\, \oplus\,\,\, \bigoplus_{j=1}^s F_{2k_j},
 \end{eqnarray*}
 where $r, s\ge 0$, $\lambda_1, \ldots, \lambda_r\in \setR$, $\sigma_1, \ldots, \sigma_s\in \setR$, 
 $\tau_{1}, \ldots, \tau_s>0$ and
 $\epsilon_1, \ldots, \epsilon_r \in \{\pm 1\}$. Moreover, 
 \begin{enumerate}
 \item
 \label{cond:simuPerA}
 $m_1 \ge \cdots\ge m_r \ge 1$ and
 $k_{1}\ge \cdots \ge k_s\ge 1$,
 \item
 \label{cond:simuPerB}
 if $m_{a} = {m_{a+1}} =  \cdots  = {m_{a+d}}$ for some $1\le a<a+d\le r$,
  then 
 $$
    \epsilon_{a} \le    \epsilon_{{a+1}} \le \cdots   \le \epsilon_{{a+d}}.
 $$
 \end{enumerate}
 \end{theorem}

\textbf{Acknowledgements.} 
The author gratefully appreciates financial
support by the Academy of Finland (project 13132527 and Centre of
Excellence in Inverse Problems Research), and by the Institute of
Mathematics at Aalto University.
I would like to thank Alberto Favaro and Tony Liimatainen for useful discussions.



\begin{thebibliography}{SWW10}

\bibitem[AMR88]{AbrahamMarsdenRatiu:1988}
R.~Abraham, J.E. Marsden, and T.~Ratiu, \emph{Manifolds, tensor analysis, and
  applications}, Springer, 1988.

\bibitem[Coh05]{Cohn:2005}
P.~M. Cohn, \emph{Basic algebra: {Groups, Rings, and Fields}}, Springer, 2005.

\bibitem[Dah09]{Dahl:2009}
M.~Dahl, \emph{Electromagnetic fields from contact- and symplectic geometry},
  preprint (2009).

\bibitem[Dah11]{Dahl:2011:Closure}
\bysame, \emph{Determining electromagnetic medium from the {Fresnel surface}},
  arXiv: 1103.3118 (2011).

\bibitem[DH80]{DruHath:1980}
I.T. Drummond and S.J. Hathrell, \emph{{QED} vacuum polarization in a
  background gravitational field and its effect on the velocity of photons},
  Physical Review D \textbf{22} (1980), no.~2, 343--355.

\bibitem[DKS89]{Dray:1989}
T.~Dray, R.~Kulkarni, and J.~Samuel, \emph{Duality and conformal structure},
  Journal of Mathematical Physics \textbf{30} (1989), no.~6, 1306--1309.

\bibitem[FB11]{FavaroBergamin:2011}
A.~Favaro and L.~Bergamin, \emph{The non-birefringent limit of all linear,
  skewonless media and its unique light-cone structure}, Annalen der Physik
  \textbf{523} (2011), no.~5, 383--401.

\bibitem[Fie86]{Fiedler:1986}
M.~Fiedler, \emph{Special matrices and their applications in numerical
  mathematics}, Kluwer, 1986.

\bibitem[Har91]{Harnett:1991}
G.~Harnett, \emph{Metrics and dual operators}, Journal of Mathematical Physics
  \textbf{32} (1991), no.~1, 84--91.

\bibitem[HO03]{Obu:2003}
F.W. Hehl and Y.N. Obukhov, \emph{Foundations of classical electrodynamics:
  Charge, flux, and metric}, Progress in Mathematical Physics, Birkh\"auser,
  2003.

\bibitem[LR05]{LancasterRodman:2005}
P.~Lancaster and L.~Rodman, \emph{Canonical forms for {Hermitian} matrix pairs
  under strict equivalence and congruence}, SIAM Review \textbf{47} (2005),
  no.~3, 407--443.

\bibitem[O'N83]{ONeill:1983}
B.~O'Neill, \emph{{Semi-Riemannian} geometry with applications to relativity},
  Academic Press, 1983.

\bibitem[PSW07]{PSW_JHEP:2007}
R.~Punzi, F.P. Schuller, and M.N.R. Wohlfarth, \emph{Area metric gravity and
  accelerating cosmology}, Journal of High Energy Physics 02 \textbf{030}
  (2007).

\bibitem[PSW09]{PunziEtAl:2009}
\bysame, \emph{Propagation of light in area metric backgrounds}, Classical and
  Quantum Gravity \textbf{26} (2009), 035024.

\bibitem[SW06]{SchullerWohlfarth:2006}
F.P. Schuller and M.N.R. Wohlfarth, \emph{Geometry of manifolds with area
  metric: multi-metric backgrounds}, Nuclear physics B \textbf{747} (2006),
  398--422.

\bibitem[SWW10]{Schuller:2010}
F.P. Schuller, C.~Witte, and M.N.R. Wohlfarth, \emph{Causal structure and
  algebraic classification of non-dissipative linear optical media}, Annals of
  Physics \textbf{325} (2010), no.~9, 1853--1883.

\end{thebibliography}

\providecommand{\bysame}{\leavevmode\hbox to3em{\hrulefill}\thinspace}
\providecommand{\MR}{\relax\ifhmode\unskip\space\fi MR }
\providecommand{\MRhref}[2]{%
  \href{http://www.ams.org/mathscinet-getitem?mr=#1}{#2}
}
\providecommand{\href}[2]{#2}

\end{document}